\RequirePackage{amsthm}%

\documentclass[sn-mathphys-num]{sn-jnl}

\usepackage{graphicx}%
\usepackage{multirow}%
\usepackage{amsmath,amssymb,amsfonts}%
\usepackage{mathrsfs}%
\usepackage[title]{appendix}%
\usepackage{xcolor}%
\usepackage{textcomp}%
\usepackage{manyfoot}%
\usepackage{booktabs}%
\usepackage{algorithm}%
\usepackage{algorithmicx}%
\usepackage{algpseudocode}%
\usepackage{listings}%

\theoremstyle{thmstyleone}%
\newtheorem{theorem}{Theorem}
\newtheorem{proposition}[theorem]{Proposition}%

\theoremstyle{thmstyletwo}%
\newtheorem{example}{Example}%

\theoremstyle{thmstylethree}%
\newtheorem{definition}{Definition}%

\usepackage{lipsum}								

\usepackage[utf8]{inputenc}                     
\usepackage[TS1,T1]{fontenc}			        

\usepackage{xpatch}								

\usepackage{booktabs}							

\usepackage{amsmath}							
\usepackage{amssymb}							
\usepackage{amsfonts}							

\usepackage[disable]{todonotes}		        

\usepackage{listings, lstautogobble}			

\usepackage[mathscr]{eucal}						

\usepackage{colortbl}							

\usepackage{amsthm}								

\usepackage{bookmark,hyperref}					

\newcommand\numberstyle[1]{%
	\footnotesize
	\color{SQLcodegray}%
	\ttfamily
	\ifnum#1<10 0\fi#1 |%
}

\makeatletter
\xpatchcmd\algorithmic
{\ALC@linenosize \arabic{ALC@line}\ALC@linenodelimiter}
{\ALC@linenosize \theALC@line\ALC@linenodelimiter}
{}{}

\newcommand\setAlgoLinenoFormat{\renewcommand*{\theALC@line}{\ifnum\value{ALC@line}<10 0\fi\arabic{ALC@line}}}
\makeatother

\let\oldalgorithmic\algorithmic
\renewcommand\algorithmic{\ttfamily\fontseries{l}\selectfont\oldalgorithmic}

\definecolor{aliceblue}{rgb}{0.94, 0.97, 1.0}
\definecolor{skyblue}{rgb}{0.53, 0.81, 0.92}
\definecolor{brightmaroon}{rgb}{0.95,0.82,0.85}

\hypersetup{%
	breaklinks=true,			
	colorlinks=true,			
	urlcolor=blue,				
	citecolor=gray,				
	linkcolor=darkgray,			
	anchorcolor=blue,			
	bookmarksopen=false,		
	linktocpage=true,			
	bookmarksnumbered,			
	pdfstartview={FitH},		
	pdftitle={Skycube émergent},
	pdfauthor={Mickaël Martin Nevot},
	pdfsubject={Article de recherche},
	pdfkeywords = {fouille de données de jeux vidéo, cubes de données fermés, cube de données émergents, Skycube, treillis de concepts, ensembles en Accord},
}

\definecolor{SQLCodeGreen}{rgb}{0,0.6,0}
\definecolor{SQLcodegray}{rgb}{0.5,0.5,0.5}
\definecolor{SQLCodePurple}{HTML}{C42043}
\definecolor{SQLBackgroundcolor}{HTML}{F2F2F2}
\definecolor{SQLBookColor}{cmyk}{0,0,0,0.90}  
\color{SQLBookColor}

\lstdefinestyle{SQLStyle} {
	backgroundcolor=\color{SQLBackgroundcolor},
	commentstyle=\color{SQLCodeGreen},
	keywordstyle=\color{SQLCodePurple},
	numberstyle=\numberstyle,
	stringstyle=\color{SQLCodePurple},
	basicstyle=\footnotesize\ttfamily,
	breakatwhitespace=false,
	breaklines=true,
	captionpos=b,
	keepspaces=true,
	numbers=left,
	numbersep=10pt,
	showspaces=false,
	showstringspaces=false,
	showtabs=false,
	autogobble=true,
	literate = 	{é}{{\'e}}{1}%
	{è}{{\`e}}{1}%
	{à}{{\`a}}{1}%
	{â}{{\^a}}{1}
	{ç}{{\c{c}}}{1}%
	{œ}{{\oe}}{1}%
	{ù}{{\`u}}{1}%
	{É}{{\'E}}{1}%
	{È}{{\`E}}{1}%
	{À}{{\`A}}{1}%
	{Ç}{{\c{C}}}{1}%
	{Œ}{{\OE}}{1}%
	{Ê}{{\^E}}{1}%
	{ê}{{\^e}}{1}%
	{î}{{\^i}}{1}%
	{ï}{{\"i}}{1}
	{ô}{{\^o}}{1}%
	{û}{{\^u}}{1}%
}

\newtheoremstyle{theoremStyle}
{9pt}{9pt}				
{}						
{}						
{\bfseries}{ -}			
{ }						
{}						

\theoremstyle{theoremStyle}

\pgfdeclarelayer{background}
\pgfdeclarelayer{foreground}
\pgfsetlayers{background, main, foreground}

\usetikzlibrary{patterns}
\usetikzlibrary{automata}

\ifdefined\proposition\else
\newtheorem{proposition}{Proposition}[section]

\ifdefined\lemma\else
\newtheorem{lemma}[proposition]{Lemme}
\fi

\ifdefined\corollaire\else
\newtheorem{corollaire}[proposition]{Corollaire}
\fi

\ifdefined\theorem\else
\newtheorem{theorem}[proposition]{Théorème}
\fi
\fi

\ifdefined\lemma\else
\newtheorem{lemma}{Lemme}
\fi

\ifdefined\corollaire\else
\newtheorem{corollaire}{Corollaire}
\fi

\ifdefined\theorem\else
\newtheorem{theorem}{Théorème}
\fi

\ifdefined\definition\else
\newtheorem{definition}{Definition}[section]
\fi

\ifdefined\example\else
\newtheorem{example}{Exemple}[section]
\fi

\ifdefined\property\else
\newtheorem{property}{Property}[section]
\fi

\ifdefined\remarques\else
\newtheorem{remarques}{remarques}
\fi

\begin{document}
	
	\title[Emerging Skycube]{Emerging Skycube\footnote{This is the author’s accepted manuscript of an article published in Knowledge and Information Systems (KAIS), Springer, 2025. The final version is available at: 10.1007/s10115-024-02320-2.}}
	
	
	\author*[1]{\fnm{Mickaël} \sur{Martin Nevot}}\email{mickael.martin-nevot@univ-amu.fr}
	
	\affil*[1]{\orgdiv{Laboratoire d'Informatique et Système}, \orgname{CRNS UMR 7020}, \orgaddress{\street{Aix-Marseille Université 52 AVENUE ESCADRILLE NORMANDIE NIEMEN}, \city{Marseille Cedex 20}, \postcode{13397}, \country{FRANCE}}}
	
	
	\abstract{Combining multi-criteria decision analysis and trend reversal discovery make it possible to extract globally optimal, or non-dominated, data in relation to several criteria, and then to observe their evolution according to a decision-making property. Thus, we introduce Emerging Skycube, a concept associating Skycube and emerging datacube. As far as we know, no DBMS-integrated solution exists to compute an emerging Skycube, and hence taking advantage of ROLAP analysis tools. An emerging datacube has only one measure: we propose to use several to comply to multi-criteria decision analysis constraints which requires multiple attributes. A datacube is expensive to compute. An emerging datacube is about twice as expensive. On the other hand, an emerging Skycube is cheaper as the trend reversal is computed after two Skycube calculations, which considerably reduces the relation volume in comparison with the initial one. It is possible to save even more computing time and storage space. To this end, we propose two successive reductions. First, a Skycube lossless partial materialisation using Skylines concepts lattice, based on the agree concepts lattice and partitions lattice. Then, either the closed emerging Skycube for an information-loss reduction, or the closed emerging L-Skycube for a smaller but lossless reduction.}

	
	
	
	\maketitle
	
	
	\section{Introduction and motivations}
	
	In this paper, we introduce the emerging Skycube and the closed emerging Skycube. The goal is to combine the multicriteria approach for preference analysis with the OLAP approach dedicated to the trend reversal discovery. The emerging Skycube extracts optimal data according to several characteristics and according to their significant evolution in regard to a property.
	
	We originally intended to reuse the emerging datacube concept (\cite{nedjarEmergingCubesTrends2007}), based on the datacube (\cite{grayDataCubeRelational1996}), by combining our formalization of the $\mathbb{C}$-IDEA closed emerging datacube algorithm, in \cite{martinnevotCIdeaFastAlgorithm2019}, with the Skycube concept (\cite{yuanEfficientComputationSkyline2005, peiCatchingBestViews2005}). We have kept our initial idea but we extended it to the whole IDEA algorithmic platform (\cite{nedjarExtractingSemanticsOLAP2011, nedjarEmergingDataCube2013}) while integrating a reduced Skycube representation based on the Skylines concepts lattice.
	
	After introducing related work, the use case we experimented on and Skyline and Skycube concepts, we present the partitions lattice, then the Agree concepts lattice, and the Skylines concepts lattice in order to propose a Skycube lossless partial materialisation. Next, we will transpose the computation of the emerging datacube, of its L border, and of the closed emerging datacube, respectively, to the emerging Skycube, to its L border and to the closed emerging Skycube.
	
	\section{Related work}
	
	The Skyline operator (\cite{borzsonySkylineOperator2001}) aim to fetch the most relevant elements in a database context. It comes from the maximal vector problem (\cite{barndorff-nielsenDistributionNumberAdmissible1966, bentleyAverageNumberMaxima1978}). However specific algorithms have been developed to comply with context specificities, with or without classical or spatial indexes, expensive to update.
	
	Furthermore, various approaches to parallelize Skyline's computation (\cite{boghWorkefficientParallelSkyline2015, lougmiriNewProgressiveMethod2017, wangScalableSpatialSkyline2018}) have been defined, especially the outstanding SkyCell (\cite{liSkyCellSpacePruningBased2021}), a parallel Skyline algorithm based on the spatial search on a discretized grid.
	
	Skyline being a decision making tool, the user will likely compute several Skylines before choosing the preferred one. To address such an issue, the Skycube concept (\cite{yuanEfficientComputationSkyline2005, peiCatchingBestViews2005}) has been proposed. As the underlying idea is to precompute all Skylines, it is essential to reduce the outcome's storage cost.
	
	A first reduction approach has been proposed by \cite{peiMultidimensionalSubspaceSkyline2006}. It was meant to determine if an item belongs to several Skycuboids. The solution proposes a Skycube representation based on a formal concept analysis where each node is a couple composed of an equivalence class (object set) and subspaces in which it belongs to the Skyline. Having the node count bounded by the tuple power set ($|\mathscr{P}(Tid(r))|$) lattice cardinality is the main drawback of this value oriented approach. As Skyline objects have priority, a large count of lattice nodes have to be accessed to rebuilt a Skycuboid. In attribute oriented approaches like ours, the node count is way smaller since it is bounded by $|\mathscr{P}(\mathcal{C})|$. 
	
	\cite{xiaRefreshingSkyCompressed2006} proposed a Skycuboid oriented approach (attribute oriented). This solution proposes a heavy reduction of the Skycube by pruning items considered as redundant from all Skycuboids. Its main drawback is the Skycuboid's computation, tricky and expansive, requiring to search through the whole data structure to fetch all redundant items. Moreover, in contrast with \cite{peiMultidimensionalSubspaceSkyline2006, peiComputingCompressedMultidimensional2007} approaches and ours, this one is not based on a theoretical concept as sound as the formal concept analysis. However, one of the main interest of this approach, besides the great storage cost reduction, is the data updating efficiency.
	
	Despite the presented approach's main goal, to reduce Skycuboids' rebuilding cost, it is still a good improvement on the data updating and storage reduction.
	
	\section{Use case and overview}
	
	In this paper, we introduce an example applied to a video game: Pokémon Showdown! and used to illustrate introduced concepts\footnote{For more information, see appendices}.
	
	\subsection{Relation example}
	
	\begin{example}
		The relation example \texttt{Pokémon} (cf. table~\ref{tab:relation_exemple_4}) preserve, from our experimentation, the \texttt{Tier} attributes of the strategic tier (cf paragraph~\ref{sssec:tier_strategique_de_pokemon_showdown}) of the play\footnote{Since it is possible to have battles with multiple tier teams, we consider the strategic tier of a battle as the lowest strategic tier of the strategic tier set played in the battle. Example, a battle between the Pokémon list $121, 113, 006 (A)$, Starmie (OU), Chansey (OU), Charizard (UU) is a strategic tier UU.}, \texttt{Player}, the Pokémon sequence played by the player, and \texttt{Opponent}, the Pokémon sequence of the opponent. The property \texttt{Rank} categorises players in two classes, newbies ($N$) and experts ($E$), depending of the Elo ranking system of the Pokémon Showdown! ladder (cf. subsubsection~\ref{ssec:ladder_de_pokemon_showdown}). The choice criteria to decide the "best apparition order for Pokémon is a battle" are Pokémon \texttt{Rarity} in the sequence, battle \texttt{Duration} (in turn count) and the \texttt{Loss} rate of the player's sequence against the opponent's sequence. The preference of all these criteria is maximum. The attribute \texttt{Rarity} is invariable from one player category to another. During data aggregation, we merge mirror battles, which are not very relevant due to their symmetrical data characteristics, and for the set, we limit to three Pokémon sequences, for both players. So as to ease representations, we show them in lists of Pokémon numbers, and we affect a letter to identify them. Therefore, the list $121, 113, 006 (A)$ identify the three-Pokémon sequence Starmie, Chansey, Charizard. For more clarity, we round the battle duration and loss rate to the nearest 5-multiple.
	\end{example}
	
	\begin{table}[htbp]
		\caption{\texttt{Pokémon} relation example}\label{tab:relation_exemple_4}
		\centering
		\begin{minipage}{\linewidth}
				\begin{tabular}{c|cccc|ccc} \toprule
					\texttt{RowId} & \texttt{Tier}\footnote{Strategic Tier.} & \texttt{Player}\footnote{With the official Pokémon numbers (and their strategic tier): \begin{itemize}\item n°006: Charizard (UU); \item n°040: Wigglytuff (UU); \item n°065: Alakazam (OU); \item n°080: Slowbro (OU); \item n°103: Exeggutor (OU); \item n°113: Chansey (OU); \item n°121: Starmie (OU); \item n°143: Snorlax (OU).\end{itemize}} & \texttt{Opponent}\footnote{\emph{idem supra}.} & \texttt{Rank}\footnote{Players category, with only two partitions considered: novice (N) and expert (E).} & \texttt{Rarity}\footnote{Let $p$ be the percentage of drop for the Pokémon sequence (which is the multiplication of the percentage of drop for each Pokémon of the sequence), the \texttt{Rarity} score $r$ is calculated, on a scale of 0 to 10, as such: $\text{if } p = 1 \text{ then } r = 0 \text{, else } r = \lfloor max(\frac{(p - 1) \times 10 - (100 - e \times 0.9)}{e}, 0) \rfloor + 1 $.} & \texttt{Duration}\footnote{In total number of turns in the fight.} & \texttt{Loss}\footnote{In percentage.} \\
					\midrule
					$1$  & $UU$ & $121, 113, 006 (A)$ & $065, 113, 143 (D)$ & $N$ & $5$ & $25$ & $30$ \\
					$2$  & $OU$ & $065, 103, 065 (B)$ & $065, 040, 065 (E)$ & $N$ & $4$ & $65$ & $50$ \\
					$3$  & $OU$ & $065, 103, 065 (B)$ & $121, 113, 121 (F)$ & $N$ & $4$ & $35$ & $40$ \\
					$4$  & $OU$ & $065, 103, 065 (B)$ & $121, 113, 006 (A)$ & $N$ & $4$ & $85$ & $40$ \\
					$5$  & $OU$ & $121, 113, 080 (C)$ & $121, 113, 006 (A)$ & $N$ & $1$ & $95$ & $60$ \\
					$6$  & $OU$ & $121, 113, 080 (C)$ & $065, 103, 065 (B)$ & $N$ & $1$ & $35$ & $50$ \\
					$7$  & $OU$ & $065, 113, 143 (D)$ & $065, 103, 065 (B)$ & $N$ & $9$ & $85$ & $60$ \\
					$8$  & $OU$ & $065, 113, 143 (D)$ & $121, 113, 080 (C)$ & $N$ & $9$ & $85$ & $70$ \\
					$9$  & $UU$ & $065, 040, 065 (E)$ & $065, 113, 143 (D)$ & $N$ & $7$ & $25$ & $50$ \\
					$10$ & $UU$ & $065, 040, 065 (E)$ & $065, 040, 065 (E)$ & $N$ & $7$ & $65$ & $30$ \\
					$11$ & $UU$ & $121, 113, 006 (A)$ & $065, 113, 143 (D)$ & $E$ & $5$ & $20$ & $30$ \\
					$12$ & $OU$ & $065, 103, 065 (B)$ & $065, 040, 065 (E)$ & $E$ & $4$ & $60$ & $45$ \\
					$13$ & $OU$ & $065, 103, 065 (B)$ & $121, 113, 121 (F)$ & $E$ & $4$ & $30$ & $30$ \\
					$14$ & $OU$ & $065, 103, 065 (B)$ & $121, 113, 006 (A)$ & $E$ & $4$ & $80$ & $50$ \\
					$15$ & $OU$ & $121, 113, 080 (C)$ & $121, 113, 006 (A)$ & $E$ & $1$ & $90$ & $70$ \\
					$16$ & $OU$ & $121, 113, 080 (C)$ & $065, 103, 065 (B)$ & $E$ & $1$ & $30$ & $30$ \\
					$17$ & $OU$ & $065, 113, 143 (D)$ & $065, 103, 065 (B)$ & $E$ & $9$ & $80$ & $50$ \\
					$18$ & $OU$ & $065, 113, 143 (D)$ & $121, 113, 080 (C)$ & $E$ & $9$ & $90$ & $70$ \\
					$19$ & $UU$ & $065, 040, 065 (E)$ & $065, 113, 143 (D)$ & $E$ & $7$ & $20$ & $30$ \\
					$20$ & $UU$ & $065, 040, 065 (E)$ & $065, 040, 065 (E)$ & $E$ & $7$ & $60$ & $45$ \\
					\bottomrule
				\end{tabular}
		\end{minipage}
	\end{table}
	
	\begin{figure}[htbp]
		\begin{minipage}{\linewidth}
			\centering
			\includegraphics[width=\linewidth]{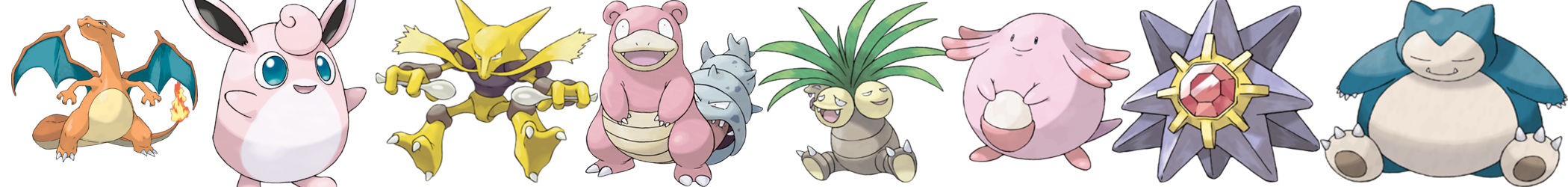}
			\caption[Charizard, Wigglytuff, Alakazam, Slowbro, Exeggutor, Chansey, Starmie, Snorlax]{The Pokémon of the use case\footnote{From left to right: Charizard (006, UU), Wigglytuff (040, UU), Alakazam (n°065, OU), Slowbro (n°080, OU), Exeggutor (n°103, OU), Chansey (n°113, OU), Starmie (n°121, OU), Snorlax (n°143, OU); illustration of Pokémon FireRed Version and LeafGreen on Poképédia. }}\label{fig:alakazam_Slowbro_Exeggutor_Chansey_Starmie_tauros_Snorlax}
		\end{minipage}
	\end{figure}
	
	\begin{table}[htbp]
		\caption{$\texttt{Pokémon}_{1}$ relation example, novice players}\label{tab:relation_exemple_4_1}
		\centering
		\begin{minipage}{\linewidth}
			\centering
				\begin{tabular}{c|cccc|ccc} \toprule
					\texttt{RowId} & \texttt{Tier}\footnote{Strategic Tier.} & \texttt{Player} & \texttt{Opponent} & \ldots & \texttt{Rarity}\footnote{Let $p$ be the percentage of drop for the Pokémon sequence (which is the multiplication of the percentage of drop for each Pokémon of the sequence), the \texttt{Rarity} score $r$ is calculated, on a scale of 0 to 10, as such: $\text{if } p = 1 \text{ then } r = 0 \text{, else } r = \lfloor max(\frac{(p - 1) \times 10 - (100 - e \times 0.9)}{e}, 0) \rfloor + 1 $.} & \texttt{Duration}\footnote{In total number of turns in the fight.} & \texttt{Loss}\footnote{In percentage.} \\
					\midrule
					$1$  & $UU$ & $A$ & $D$ & \ldots & $5$ & $25$ & $30$ \\
					$2$  & $OU$ & $B$ & $E$ & \ldots & $4$ & $65$ & $50$ \\
					$3$  & $OU$ & $B$ & $F$ & \ldots & $4$ & $35$ & $40$ \\
					$4$  & $OU$ & $B$ & $A$ & \ldots & $4$ & $85$ & $40$ \\
					$5$  & $OU$ & $C$ & $A$ & \ldots & $1$ & $95$ & $60$ \\
					$6$  & $OU$ & $C$ & $B$ & \ldots & $1$ & $35$ & $50$ \\
					$7$  & $OU$ & $D$ & $B$ & \ldots & $9$ & $85$ & $60$ \\
					$8$  & $OU$ & $D$ & $C$ & \ldots & $9$ & $85$ & $70$ \\
					$9$  & $UU$ & $E$ & $D$ & \ldots & $7$ & $25$ & $50$ \\
					$10$ & $UU$ & $E$ & $E$ & \ldots & $7$ & $65$ & $30$ \\
					\bottomrule
				\end{tabular}
		\end{minipage}
	\end{table}
	
	\begin{table}[htbp]
		\caption{$\texttt{Pokémon}_{2}$ relation example, expert players}\label{tab:relation_exemple_4_2}
		\centering
		\begin{minipage}{\linewidth}
			\centering
				\begin{tabular}{c|cccc|ccc} \toprule
					\texttt{RowId} & \texttt{Tier}\footnote{Tier strategic.} & \texttt{Player} & \texttt{Opponent} & \ldots & \texttt{Rarity}\footnote{Let $p$ be the percentage of drop for the Pokémon sequence (which is the multiplication of the percentage of drop for each Pokémon of the sequence), the \texttt{Rarity} score $r$ is calculated, on a scale of 0 to 10, as such: $\text{if } p = 1 \text{ then } r = 0 \text{, else } r = \lfloor max(\frac{(p - 1) \times 10 - (100 - e \times 0.9)}{e}, 0) \rfloor + 1 $.} & \texttt{Duration}\footnote{In total number of turns in the fight.} & \texttt{Loss}\footnote{In percentage.} \\
					\midrule
					$1$  & $UU$ & $A$ & $D$ & \ldots & $5$ & $20$ & $30$ \\
					$2$  & $OU$ & $B$ & $E$ & \ldots & $4$ & $60$ & $45$ \\
					$3$  & $OU$ & $B$ & $F$ & \ldots & $4$ & $30$ & $30$ \\
					$4$  & $OU$ & $B$ & $A$ & \ldots & $4$ & $80$ & $50$ \\
					$5$  & $OU$ & $C$ & $A$ & \ldots & $1$ & $90$ & $70$ \\
					$6$  & $OU$ & $C$ & $B$ & \ldots & $1$ & $30$ & $30$ \\
					$7$  & $OU$ & $D$ & $B$ & \ldots & $9$ & $80$ & $50$ \\
					$8$  & $OU$ & $D$ & $C$ & \ldots & $9$ & $90$ & $70$ \\
					$9$  & $UU$ & $E$ & $D$ & \ldots & $7$ & $20$ & $30$ \\
					$10$ & $UU$ & $E$ & $E$ & \ldots & $7$ & $60$ & $45$ \\
					\bottomrule
				\end{tabular}
		\end{minipage}
	\end{table}
	
	\begin{example}
		Assuming we are looking for trend reversals occuring between \texttt{Rank} novice players' strategies and those of \texttt{Rank} expert players of the \texttt{Pokémon} relation (cf. table~\ref{tab:relation_exemple_4}). the $\texttt{Pokémon}_{1}$ relation (cf. table~\ref{tab:relation_exemple_4_1}) shows the novice players and $\texttt{Pokémon}_{2}$ (cf. table~\ref{tab:relation_exemple_4_2}) fetches tuples from \texttt{Pokémon} where $\texttt{Rank} = 'E'$. The deduplicated merged relation $\texttt{Pokémon}_{Mer}^+$ (cf. table~\ref{tab:relation_fusionnee_dedoublonnee_pokemon}) is given, for now, as a rough guide (with \texttt{R} for \texttt{Rarity}, \texttt{D}$_1$ for \texttt{Duration}$_1$, \texttt{L}$_1$ for \texttt{Loss}$_1$, \texttt{D}$_2$ for \texttt{Duration}$_2$ and \texttt{L}$_2$) for \texttt{Loss}$_2$.
	\end{example}
	
	\begin{table}[htbp]
		\caption{The deduplicated merged relation $\texttt{Pokémon}_{Mer}^+$}\label{tab:relation_fusionnee_dedoublonnee_pokemon}
		\centering
		\begin{minipage}{\linewidth}
			\centering
				\begin{tabular}{c|cccc|ccccc} \toprule
					\texttt{RowId} & \texttt{Tier}\footnote{Tier strategic.} & \texttt{Player} & \texttt{Opponent} & \ldots & \texttt{R}\footnote{Let $p$ be the percentage of drop for the Pokémon sequence (which is the multiplication of the percentage of drop for each Pokémon of the sequence), the \texttt{R} score $r$ is calculated, on a scale of 0 to 10, as such: $\text{if } p = 1 \text{ then } r = 0 \text{, else } r = \lfloor max(\frac{(p - 1) \times 10 - (100 - e \times 0.9)}{e}, 0) \rfloor + 1 $.} & \texttt{D}$_1$\footnote{In total number of turns in the fight.} & \texttt{L}$_1$\footnote{In percentage.} & \texttt{D}$_2$\footnote{\emph{idem} \texttt{D}$_1$.} & \texttt{L}$_2$\footnote{\emph{idem} \texttt{L}$_1$.} \\
					\midrule
					$1$  & $UU$ & $A$ & $D$ & \ldots & $5$ & $25$ & $30$ & $20$ & $30$ \\
					$2$  & $OU$ & $B$ & $E$ & \ldots & $4$ & $65$ & $50$ & $60$ & $45$ \\
					$3$  & $OU$ & $B$ & $F$ & \ldots & $4$ & $35$ & $40$ & $30$ & $30$ \\
					$4$  & $OU$ & $B$ & $A$ & \ldots & $4$ & $85$ & $40$ & $80$ & $50$ \\
					$5$  & $OU$ & $C$ & $A$ & \ldots & $1$ & $95$ & $60$ & $90$ & $70$ \\
					$6$  & $OU$ & $C$ & $B$ & \ldots & $1$ & $35$ & $50$ & $30$ & $30$ \\
					$7$  & $OU$ & $D$ & $B$ & \ldots & $9$ & $85$ & $60$ & $80$ & $50$ \\
					$8$  & $OU$ & $D$ & $C$ & \ldots & $9$ & $85$ & $70$ & $90$ & $70$ \\
					$9$  & $UU$ & $E$ & $D$ & \ldots & $7$ & $25$ & $50$ & $20$ & $30$ \\
					$10$ & $UU$ & $E$ & $E$ & \ldots & $7$ & $65$ & $30$ & $60$ & $45$ \\
					\bottomrule
				\end{tabular}
		\end{minipage}
	\end{table}
	
	\begin{example}
		The deduplicated merged relation of the example relation is given in the table~\ref{tab:relation_fusionnee_dedoublonnee_pokemon} (with \texttt{R} for \texttt{Rarity}, \texttt{D}$_1$ and \texttt{L}$_1$ respectively matching attributes \texttt{Duration} and \texttt{Loss} of novice players and \texttt{D}$_2$ et \texttt{L}$_2$ respectively matching attributes \texttt{Duration} et \texttt{Loss} of expert players).
	\end{example}
	
	A Skycube can be represented as a lattice like the one used for the novice players datacube (cf. figure~\ref{fig:skycube_treillis_complet_3_1}) or the expert one (cf. figure~\ref{fig:skycube_treillis_complet_3_2}). The Skycube cuboids, called Skycuboids, are grouped by their criteria count, by level, from the bottom, Skycuboids with only one criteria (ignoring $\emptyset$), to the top, presenting the Skycuboid with all the considered criteria.
	
	\begin{figure}[htbp]
		\centering
		\tiny
		\resizebox{1\textwidth}{!}{
			\begin{tikzpicture}[
				line join=bevel
				]
				
				\node (bottom) at (150pt, 0pt) {$\emptyset$};
				\node (n11) at (50pt, 75pt)
				{
					\begin{tabular}{c|c}
						\toprule
						\texttt{Id} & \texttt{R} \\
						\midrule
						$5$ & $1$ \\
						$6$ & $1$ \\
						\bottomrule
					\end{tabular}
				};
				\node (n12) at (150pt, 75pt)
				{
					\begin{tabular}{c|c}
						\toprule
						\texttt{Id} & \texttt{D} \\
						\midrule
						$1$ & $25$ \\
						$9$ & $25$ \\
						\bottomrule
					\end{tabular}
				};
				\node (n13) at (250pt, 75pt)
				{
					\begin{tabular}{c|c}
						\toprule
						\texttt{Id} & \texttt{L} \\
						\midrule
						$1$ & $30$ \\
						$10$ & $30$ \\
						\bottomrule
					\end{tabular}
				};
				\node (n21) at (50pt, 150pt)
				{
					\begin{tabular}{c|cc}
						\toprule
						\texttt{Id} & \texttt{R} & \texttt{D} \\
						\midrule
						$1$ & $5$ & $25$ \\
						$6$ & $1$ & $35$ \\
						\bottomrule
					\end{tabular}
				};
				\node (n22) at (150pt, 150pt)
				{
					\begin{tabular}{c|cc}
						\toprule
						\texttt{Id} & \texttt{R} & \texttt{L} \\
						\midrule
						$1$ & $5$ & $30$ \\
						$3$ & $4$ & $40$ \\
						$4$ & $4$ & $40$ \\
						$6$ & $1$ & $50$ \\
						\bottomrule
					\end{tabular}
				};
				\node (n23) at (250pt, 150pt)
				{
					\begin{tabular}{c|cc}
						\toprule
						\texttt{Id} & \texttt{D} & \texttt{L} \\
						\midrule
						$1$ & $25$ & $30$ \\
						\bottomrule
					\end{tabular}
				};
				\node (top) at (150pt, 225pt)
				{
					\begin{tabular}{c|ccc}
						\toprule
						\texttt{Id} & \texttt{R} & \texttt{D} & \texttt{L} \\
						\midrule
						$1$ & $5$ & $25$ & $30$ \\
						$3$ & $4$ & $35$ & $40$ \\
						$6$ & $1$ & $35$ & $50$ \\
						\bottomrule
					\end{tabular}
				};
				
				\draw [stealth-] (bottom) -- (n11.south);
				\draw [stealth-] (bottom) -- (n12.south);
				\draw [stealth-] (bottom) -- (n13.south); 
				\draw [stealth-] (n11.90) -- (n21.270);
				\draw [stealth-] (n11.65) -- (n22.245);
				
				\draw [stealth-] (n12.115) -- (n21.290);
				\draw [stealth-] (n12.65) -- (n23.245);
				
				\draw [stealth-] (n13.115) -- (n22.290);
				\draw [stealth-] (n13.90) -- (n23.270);
				\draw [stealth-] (n21.north) -- (top);
				\draw [stealth-] (n22.north) -- (top);
				\draw [stealth-] (n23.north) -- (top);
			\end{tikzpicture}
		}
		\caption{Representation of the $\texttt{Pokémon}_{1}$ relation's Skycube lattice}\label{fig:skycube_treillis_complet_3_1}
	\end{figure}
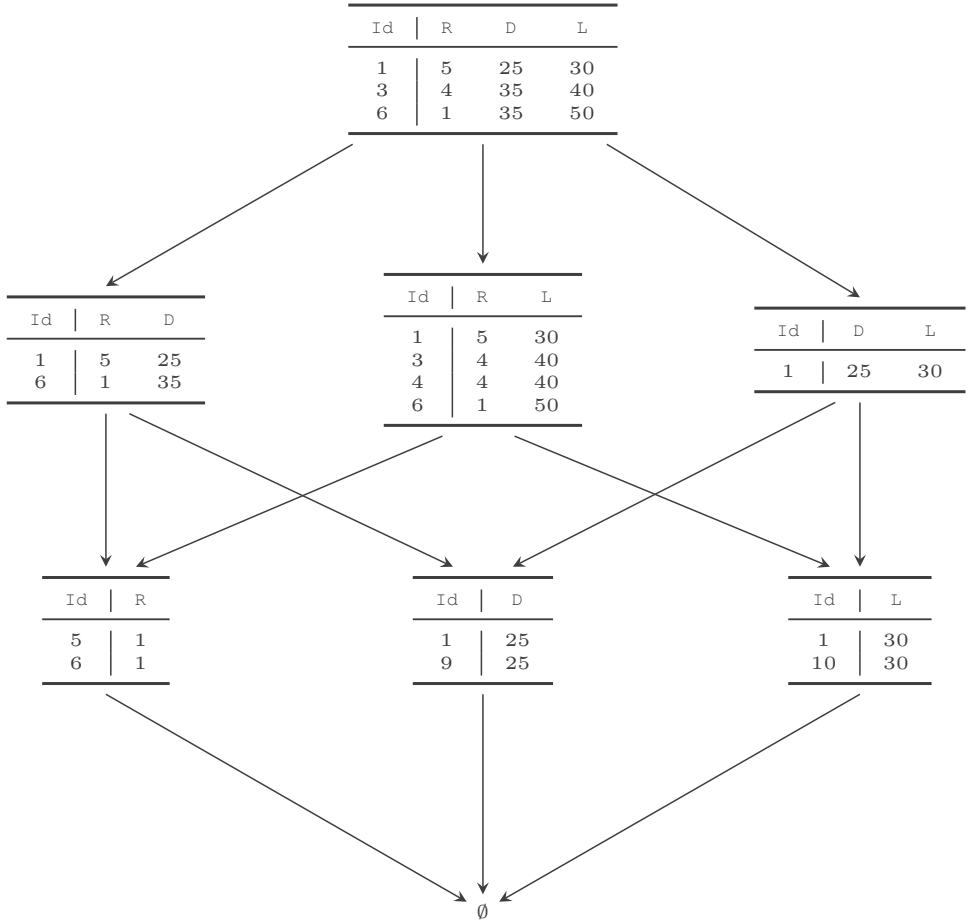
	
	\begin{figure}[htbp]
		\centering
		\tiny
		\resizebox{1\textwidth}{!}{
			\begin{tikzpicture}[
				line join=bevel
				]
				
				\node (bottom) at (150pt, 0pt) {$\emptyset$};
				\node (n11) at (50pt, 75pt)
				{
					\begin{tabular}{c|c}
						\toprule
						\texttt{Id} & \texttt{R} \\
						\midrule
						$5$ & $1$ \\
						$6$ & $1$ \\
						\bottomrule
					\end{tabular}
				};
				\node (n12) at (150pt, 75pt)
				{
					\begin{tabular}{c|c}
						\toprule
						\texttt{Id} & \texttt{D} \\
						\midrule
						$1$ & $20$ \\
						$9$ & $20$ \\
						\bottomrule
					\end{tabular}
				};
				\node (n13) at (250pt, 75pt)
				{
					\begin{tabular}{c|c}
						\toprule
						\texttt{Id} & \texttt{L} \\
						\midrule
						$1$ & $30$ \\
						$3$ & $30$ \\
						$6$ & $30$ \\
						$9$ & $30$ \\
						\bottomrule
					\end{tabular}
				};
				\node (n21) at (50pt, 150pt)
				{
					\begin{tabular}{c|cc}
						\toprule
						\texttt{Id} & \texttt{R} & \texttt{D} \\
						\midrule
						$1$ & $5$ & $20$ \\
						$6$ & $1$ & $30$ \\
						\bottomrule
					\end{tabular}
				};
				\node (n22) at (150pt, 150pt)
				{
					\begin{tabular}{c|cc}
						\toprule
						\texttt{Id} & \texttt{R} & \texttt{L} \\
						\midrule
						$6$ & $1$ & $30$ \\
						\bottomrule
					\end{tabular}
				};
				\node (n23) at (250pt, 150pt)
				{
					\begin{tabular}{c|cc}
						\toprule
						\texttt{Id} & \texttt{D} & \texttt{L} \\
						\midrule
						$1$ & $20$ & $30$ \\
						$9$ & $20$ & $30$ \\
						\bottomrule
					\end{tabular}
				};
				\node (top) at (150pt, 225pt)
				{
					\begin{tabular}{c|ccc}
						\toprule
						\texttt{Id} & \texttt{R} & \texttt{D} & \texttt{L} \\
						\midrule
						$1$ & $5$ & $20$ & $30$ \\
						$6$ & $3$ & $30$ & $30$ \\
						\bottomrule
					\end{tabular}
				};
				
				\draw [stealth-] (bottom) -- (n11.south);
				\draw [stealth-] (bottom) -- (n12.south);
				\draw [stealth-] (bottom) -- (n13.south); 
				\draw [stealth-] (n11.90) -- (n21.270);
				\draw [stealth-] (n11.65) -- (n22.245);
				
				\draw [stealth-] (n12.115) -- (n21.290);
				\draw [stealth-] (n12.65) -- (n23.245);
				
				\draw [stealth-] (n13.115) -- (n22.290);
				\draw [stealth-] (n13.90) -- (n23.270);
				\draw [stealth-] (n21.north) -- (top);
				\draw [stealth-] (n22.north) -- (top);
				\draw [stealth-] (n23.north) -- (top);
			\end{tikzpicture}
		}
		\caption{Representation of the $\texttt{Pokémon}_{2}$ relation's Skycube lattice}\label{fig:skycube_treillis_complet_3_2}
	\end{figure}
	
	\begin{example}
		The figure~\ref{fig:skycube_treillis_complet_3_1} shows the Skycube associated to the example relation $\texttt{Pokémon}_{1}$ (cf. table~\ref{tab:relation_exemple_4_1}) for novice \texttt{Rank} players, and the figure~\ref{fig:skycube_treillis_complet_3_2} shows the Skycube associated to the example relation $\texttt{Pokémon}_{2}$ (cf. table~\ref{tab:relation_exemple_4_2}) for expert players. Criteria are symbolized by their initials (\texttt{R} for \texttt{Rarity}, \texttt{D} for \texttt{Duration} and \texttt{L} for \texttt{Loss}).
	\end{example}
	
	\subsection{Overview of the proposed method}\label{ssec:apercu_de_la_methode_proposee}
	
	The main steps of the proposed method are summarize as follow:
	\begin{enumerate}
		\item Skycubes lattice computation
		\item Agree concepts lattices computation
		\item Skycubes lattice reduced representation
		\item Merged relation
		\item Abridged merged relation
		\item Emerging Skycube of the Abridged merged relation
		\item Emerging Skycube with or without loss reductions
	\end{enumerate}
	
	We detail these different steps in this paper.
	
	\section{Skyline and Skycube}\label{sec:skyline_et_skycube}
	
	In a decision-making context, some queries return no results. In these queries, the user looks for tuples with optimal values on some criteria. Using several of these criteria makes these queries unsuccessful. In fact, a tuple may be optimal for a criterion but not for another, so it is removed from the result even though it could have been relevant for the user.
	
	\subsection{Skyline}\label{ssec:skyline}
	
	In order to provide an adequate answer to the described queries, the Skyline operator (\cite{borzsonySkylineOperator2001}), previously known as Pareto set and as maximal vector problem (\cite{bentleyAverageNumberMaxima1978}), was introduced. It consider the whole set of criteria of a search as that much preferences and extracts globally optimal tuples for this preferences set. Therefore, rather than searching an hypothetical ideal solution, it extracts the possible candidates closest to the user wishes. Its general principle is based on the concept of dominance. An object or a tuple is said dominated by another if, for all criteria relevant to the decision maker, it is less optimal than this other. Such tuple is eliminated from the result, not because it is irrelevant for one of the criteria but because it is not optimal for this combination of criteria. In other words, for the user, there is at least one better solution, which will be kept.
	
	\subsection{Skycube}\label{ssec:skycube}
	
	A multidimensional generalization of the Skyline operator has been proposed through Skycube (\cite{yuanEfficientComputationSkyline2005, peiCatchingBestViews2005}). This structure brings together all the possible Skylines for all different combination of criteria. It is then possible to search efficiently dominant entities based on different criteria combination. Moreover, with this structure, it becomes possible to observe the behavior of dominant entities across the dimensional space and thus to analyze and to understand the different dominance factors. As this concept is inspired by datacube, it suffers from the same inconveniences of computational cost and storage space explosion. So, as for the datacube, it is natural to try to propose reduced representations and associated algorithms.
	
	\begin{definition}[Skyline subspace]\index{Skyline subspace}
		A subset of dimensions $\mathcal{C} \subseteq \mathcal{D}$ ($B \neq \emptyset$) forms a subspace $|\mathcal{C}|$-dimensional of $\mathcal{D}$. For a tuple $t$ in the $\mathcal{D}$ space, the projection of $t$ in the subspace $\mathcal{C}$, noted $t[\mathcal{C}]$, is a $|\mathcal{C}|$-tuple $(t.Di_1, \dotsc, t.Di_{|\mathcal{C}|})$, with $t.Di_1, \dotsc, t.Di_{|\mathcal{C}|} \in \mathcal{C}$ and $i_1 < \dotsc < i_{|\mathcal{C}|}$. The projection of a tuple $t$ ($t \in r$) in a subspace $\mathcal{C} \subseteq \mathcal{D}$ belongs to the Skyline subspace according to $\mathcal{C}$, if no tuple $t'[\mathcal{C}]$ (with $t' \in r$) dominates $t[\mathcal{C}]$ in $\mathcal{C}$. $t$ is called an object of the Skyline subspace according to $\mathcal{C}$. We call $SKY_{\mathcal{C}}(r)$ the Skyline subspace according to $\mathcal{C}$ for the $r$ relation.
	\end{definition}
	
	\begin{definition}[Skycube]\index{Skycube}
		A Skycube is the set of all Skylines in all the non-empty possible subspaces of $\mathcal{C}$: 
		\[S_C(r, \mathcal{C}) = \{ (C, SKY_{C}(r)) \mid C \subseteq \mathcal{C}\}\]
		
		$SKY_{C}(r)$ is called the Skyline cuboid (or Skycuboid) of the subspace $C$. By convention the Skycuboid of the empty criteria set is empty (\emph{i.e.} $SKY_{\emptyset}(r) = \emptyset$).
	\end{definition}
	
	The Skycube structure can be represented by a lattice similar to the datacube's lattice. Skycuboids are grouped by level depending on their criteria count. These levels are numbered from the bottom of the lattice (Skycuboids on a single criterion) and to the top (Skycuboid on all possible criteria).
	
	A multidimensional Skyline query returns the subset of tuples from the original relation forming the Skyline in a given subspace. Once the Skyline calculated, any query can be answered efficiently.
	
	Skyline belonging is not monotonous, that is a tuple $t$ belonging to a Skycuboid $SKY_{\mathcal{U}}(r)$ is not automatically contained in this Skycuboid's ancestors .
	
	\section{Agree concept and Skyline concept lattices}\label{sec:treillis_des_concepts_accord_et_skylines}
	
	As for the datacube, the Skycube can contain superfluous information. It is this problematic that motivated the proposal of Skycube reduced representations (\cite{peiMultidimensionalSubspaceSkyline2006, liuMulticonstraintShortestPath2023}). Our contribution addresses the same reduction problematic by also combining formal concept analysis (\cite{ganterFormalConceptAnalysis2024}) and the Skyline. To prevent the important cost to rebuild Skycuboids induced by value-oriented grouping from \cite{peiMultidimensionalSubspaceSkyline2006}, our reduction method chooses a criterion-oriented grouping approach based on agree sets.
	
	In this section, our goal is to define a formal framework combining the concept of agree set and the concept lattice. We propose a new structure, the Agree concept lattice of a relation, on which our partial Skycube materialization is based. After a reminder of the notions of partition lattice, agree set, associated equivalence classes and the Agree concept lattice (\cite{lakhalMultidimensionalSkylineAnalysis2017}), we characterize the Skyline concept lattice for emerging Skycube partial materialization.
	
	\subsection{The partition lattice}
	
	The concepts reminded in this subsection have been used for solving database problems (\cite{spyratosPartitionModelDeductive1987}).
	
	\begin{definition}[Partition of a set]\index{Partition of a set}
		Let $E$ be a set, a partition $\pi(E)$\footnote{When there is no ambiguity for a set $E$, we note the partition $\pi$.} of a set $E$ is a family of subsets of this set such as each item of $E$ exactly belongs to only one of these families (or classes). In other words, $\pi(E)$ is a family of disjoint sets ($\forall X, Y \in \pi(E)$ we have $X \cap Y = \emptyset$) and their union is equal to $E$ ($\bigcup_{X \in \pi(E)} = E$).
	\end{definition}
	
	\begin{definition}[Order relationship between partitions]\index{Order relationship between partitions}
		Let $\pi(E)$, $\pi'(E)$ be two partitions of a set $E$, $\pi(E)$ is refinement of $\pi'(E)$ if and only if any class of $\pi(E)$ is obtained by dividing classes of $\pi'(E)$\footnote{In a equivalent way, $\pi(E)$ is a refinement of $\pi'(E)$ if and only if any class of $\pi'(E)$ is the result from the union of classes of $\pi(E)$.}. The refinement relationship between two partitions is an partial order relationship noted $\sqsubseteq$. It is defined as:
		\[\pi(E) \sqsubseteq \pi'(E) \Leftrightarrow \pi(E) \text{ is a refinement of } \pi'(E)\footnote{Reciprocally $\pi'(E)$ is said rougher than $\pi(E)$.} \Leftrightarrow (\forall X \in \pi(E), \exists X' \in \pi'(E), X \subseteq X')\]
	\end{definition}
	
	\begin{definition}[Partition product]\index{Partition product}
		Let $\pi(E)$ and $\pi'(E)$ be two partitions of a set $E$. The partition product of $\pi(E)$ and $\pi'(E)$, noted $\pi(E) \bullet \pi'(E)$, is obtained as:
		\[\pi(E) \bullet \pi'(E) = \{ Z = X \cap Y \mid Z \neq \emptyset, X \in \pi(E) \text{ and } Y \in \pi'(E)\}\]
	\end{definition}
	
	\begin{definition}[Sum of partitions]\index{Sum of partitions}
		Let the helper function $R$ be:
		\[R(e, F) = \bigcup_{\substack{X \in F \\ e \in X}} X\]
		
		With $e$ an item of a set $E$, $F$ a family of subsets of $E$. $R(e, F)$ is the union of $F$ sets containing $e$.
		
		Let $\pi(E)$ and $\pi'(E)$ be two partitions of a set $E$. The sum of the partitions $\pi(E)$ and $\pi'(E)$, noted $\pi(E) + \pi'(E)$, is obtained by transitive closure of the operation which associates an item of $E$ to the set of elements of its classes in $\pi(E)$ and $\pi'(E)$ (\cite{birkhoffLatticeTheory1967}). The sequence $S$ is defined below to formalize this computation: 
		\[
		\left\{
		\begin{array}{lcl}
			S_0 = \max_{\subseteq}({\pi(E) \cup \pi'(E)}) \\
			S_n = \max_{\subseteq}(\{ R(e, S_{n-1}) \mid e \in E\})
		\end{array}\right.
		\]
		
		So the sum operator can be defined as:
		\[\pi(E) + \pi'(E) = S_k \text{ with } k \text{ as } S_k = S_{k-1}\]
	\end{definition}
	
	\begin{theorem}[Partition lattice]\index{Partition lattice}
		Let $\Pi(E)$ be the set of possible partitions of a set $E$. The ordered set $\langle\Pi(E), \sqsubseteq \rangle$ forms a complete lattice named partition lattice of $E$. $\forall P \subseteq \Pi(E)$, its \emph{infimum} or lower bound ($\bigwedge$) and its \emph{supremum} or upper bound ($\bigvee$) are given below:
		\begin{align*}
			\bigwedge P &= \bullet_{\pi \in P} \pi \\
			\bigvee P &= +_{\pi \in P} \pi
		\end{align*}
	\end{theorem}
	
	\begin{example}
		The Hasse diagram of the partition lattice of the set $E = \{1, 2, 3, 4\}$ is shown in figure~\ref{treillis_des_partitions}. The roughest partitions are at the bottom and the thinnest ones at the top.
	\end{example}
	
	\begin{figure}
		\centering
		\resizebox{1\textwidth}{!}{
			\begin{tikzpicture}[
				line join=bevel,
				]
				
				\node (bottom) at (150pt, 0pt) {$\{1234\}$};
				\node (p11) at (0pt, 75pt) {$\{13,24\}$};  
				\node (p12) at (50pt, 75pt) {$\{1,234\}$};
				\node (p13) at (100pt, 75pt) {$\{2,134\}$};
				\node (p14) at (150pt, 75pt) {$\{3,124\}$};
				\node (p15) at (200pt, 75pt) {$\{12,34\}$};    
				\node (p16) at (250pt, 75pt) {$\{4,123\}$};
				\node (p17) at (300pt, 75pt) {$\{14,23\}$};
				\node (p21) at (0pt, 150pt) {$\{1,3,24\}$};
				\node (p22) at (60pt, 150pt) {$\{1,2,34\}$};
				\node (p23) at (120pt, 150pt) {$\{2,4,13\}$};
				\node (p24) at (180pt, 150pt) {$\{2,3,14\}$};
				\node (p25) at (240pt, 150pt) {$\{3,4,12\}$};        
				\node (p26) at (300pt, 150pt) {$\{1,4,23\}$};
				\node (top) at (150pt, 225pt) {$\{1,2,3,4\}$};
				
				\draw [] (bottom) -- (p11.south);
				\draw [] (bottom) -- (p12.south);
				\draw [] (bottom) -- (p13.south);
				\draw [] (bottom) -- (p14.south);
				\draw [] (bottom) -- (p15.south);
				\draw [] (bottom) -- (p16.south);
				\draw [] (bottom) -- (p17.south);
				\draw [] (p11.90) -- (p21.south);
				\draw [] (p11.65) -- (p23.south);
				
				\draw [] (p12.115) -- (p21.south);
				\draw [] (p12.90) -- (p22.south);
				\draw [] (p12.65) -- (p26.south);
				
				\draw [] (p13.115) -- (p22.south);
				\draw [] (p13.90) -- (p23.south);
				\draw [] (p13.65) -- (p24.south);
				
				\draw [] (p14.115) -- (p21.south);
				\draw [] (p14.90) -- (p24.south);
				\draw [] (p14.65) -- (p25.south);
				
				\draw [] (p15.90) -- (p22.south);
				\draw [] (p15.65) -- (p25.south);
				
				\draw [] (p16.115) -- (p23.south);
				\draw [] (p16.90) -- (p25.south);
				\draw [] (p16.65) -- (p26.south);
				
				\draw [] (p17.90) -- (p24.south);
				\draw [] (p17.65) -- (p26.south);
				\draw [] (p21.north) -- (top);
				\draw [] (p22.north) -- (top);
				\draw [] (p23.north) -- (top);
				\draw [] (p24.north) -- (top);
				\draw [] (p25.north) -- (top);
				\draw [] (p26.north) -- (top);
			\end{tikzpicture}
		}
		\caption{Hasse diagram of partition lattice of $E = \{1, 2, 3, 4\}$}\label{treillis_des_partitions}
	\end{figure}
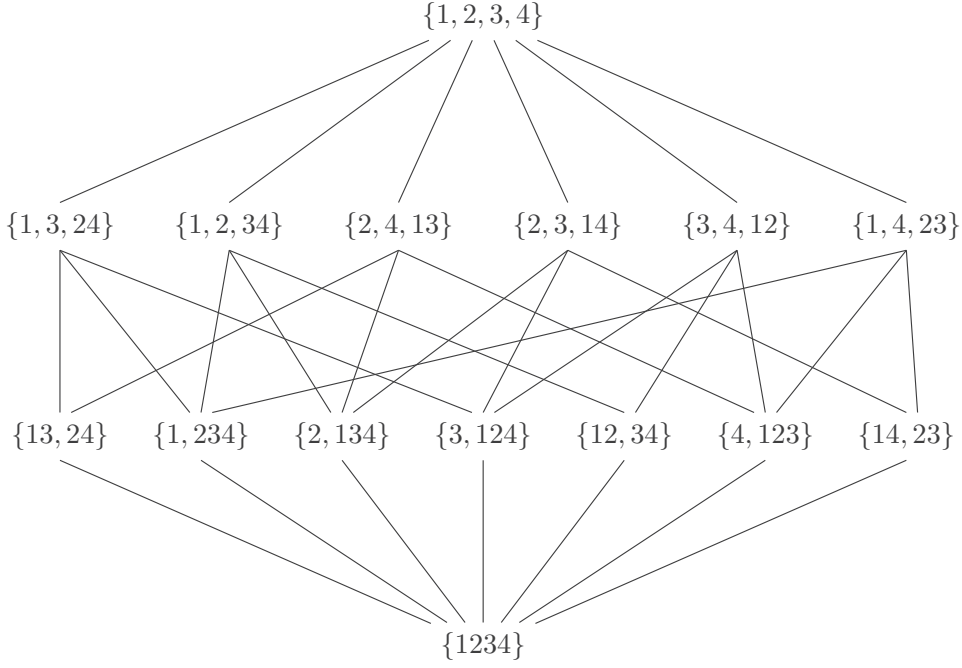
	
	\subsection{Agree sets}
	
	The concept of agree set, as well as its associated closure system, introduced in \cite{beeriStructureArmstrongRelations1984} to characterize the Armstrong relation, has been successfully used to discover exact and approximate functional dependencies (\cite{lopesFunctionalApproximateDependency2002}). Two tuples agree on an attribute set $X$ if they share the same values on $X$.
	
	We begin by presenting the agree set definition, then we give other derived definitions.
	
	\begin{definition}[Agree set]\index{Agree set}
		Let $P \subseteq r$ be a subset of $r$ relation's tuples and $D \subseteq \mathcal{D}$ a dimension set, $P$ agree on $D$ if and only if $\forall t, t' \in P$ we have $t[D] = t'[D]$.
		
		In other words, let $t$, $t'$ be two tuples of $r$ and $X \subseteq \mathcal{C}$ an attribute set (criteria in our context). $t$, $t'$ agree on $X$ if and only if $t_i[X] = t_j[X]$. The agree set of $t$ and $t'$, $Agr(t, t')$ is defined as follows:
		\[Agr(t, t') = \{C \in \mathcal{C} \mid t[C] = t'[C]\}\]
		
		This definition can be generalized to a tuples set $T \subseteq r$ composed by at least two elements:
		\[Agr(T) = \{C \in \mathcal{C} \mid t[C] = t'[C], \forall t, t' \in T \}\]
	\end{definition}
	
	\begin{definition}[Database relation agree set]\index{Database relation agree set}
		The agree attribute set $Agree$ of a relation $r$ is defined as follows: 
		\[Agree(r) = \{Agr(t_i, t_j) \mid t_i, t_j \in r \text{ and } i \neq j\}\]
		
		This set can be redefined in a equivalent way:
		\[Agree(r) = \{Agr(T) \mid \forall T \subseteq r \text{ and } |T| \geq 2\}\] 
	\end{definition}
	
	\begin{example}
		With the $\texttt{Pokémon}_{1}$ relation (cf. table~\ref{tab:relation_exemple_4_1}), $Agr(t_7, t_8) = RD$ because these two tuples share the same value on \texttt{Rarity} and \texttt{Duration}, and have different values on \texttt{Loss}. Also, $Agr(\{t_2, t_6, t_9\}) = L$ because these three tuples share the same value only on the criterion \texttt{Loss}. The agree attribute set of $\texttt{Pokémon}_{1}$ relation is: $Agree(\texttt{Pokémon}_{1}) = \{\emptyset, R, D, L, RD, RL\}$.
	\end{example}
	
	\begin{definition}[Tuple equivalence class]\index{Tuple equivalence class}
		Let $r$ be a relation and $C \subseteq \mathcal{C}$ a criterion set. The tuple $t$ equivalence class according to $C$, noted $[t]_C$, is defined as the identifier $i$ (\texttt{RowId}) set of all the tuples $t_i\in r $ which agree on $t$ according to $C$ (\emph{i.e.} the identifier set of tuples $t_i$ sharing with $t$ the same values for $C$). So we have:
		\[[t]_C = \{i \in Tid(r) \mid t_i[C] = t[C]\}\]
	\end{definition}
	
	\begin{example}
		With the $\texttt{Pokémon}_{1}$ relation, $[t_2]_{R} = \{2, 3, 4\}$ because only the tuples $t_2$, $t_3$ and $t_4$ share the same value on criterion \texttt{Rarity}.
	\end{example}
	
	\subsection{Relation agree concepts}
	
	Our objective is to define a specific concept lattice based on agree sets and database partitions (\cite{spyratosPartitionModelDeductive1987}). Thus, we characterize an instance of the Galois connection between the power set (set of all subsets) lattice of the criterion set and the partitions lattice of the tuple identifier set. This connection enables the possibility to define dual closure operators, to introduce the Agree concept and to characterize the Agree concept lattice.
	
	\begin{definition}
		Let $\texttt{RowId}: r \rightarrow \mathbb{N}$ be an application associating to each tuple a unique natural integer and $Tid(r) = \{Rowid(t) \mid t \in r\}$ the tuple identifier set. Let $f$, $g$ be two applications between the ordered sets $\langle\Pi(Tid(r)), \sqsubseteq \rangle$ and $\langle\mathscr{P}(\mathcal{C}), \subseteq\rangle$ defined as follows:
		\[
		\begin{array}{lrcl}
			f: & \langle\Pi(Tid(r)), \sqsubseteq \rangle & \longrightarrow & \langle\mathscr{P}(\mathcal{C}), \subseteq\rangle \\
			& \pi & \longmapsto & \displaystyle\bigcap_{[t] \in \pi} Agr(\{t_i \mid i \in [t]\}) \\
			g: & \langle\mathscr{P}(\mathcal{C}), \subseteq\rangle & \longrightarrow & \langle\Pi(Tid(r)), \sqsubseteq \rangle \\
			& C & \longmapsto & \displaystyle \{[t]_{C} \mid t \in r \}
		\end{array}
		\]
	\end{definition}
	
	For a criterion set $C$, the equivalence class set according to $C$ forms a partition of $Tid(r)$. The application $g$ associates $C$ to this identifier partition. The latter is noted $\pi_C$ and defined as $\pi_C = g(C)$. The set of all the possible partitions $\pi_C$ is noted $\Pi_{\mathscr{P}(\mathcal{C})}$. The application $f$ handles the opposite association.
	
	\begin{example}\label{ex:accords_f_g}
		With the $\texttt{Pokémon}_{2}$ relation, assuming the criterion sets $R$, $D$ and $DL$, we have $g(R) = \{1, 234, 56, 78, 9[10]\}$ (are in the same equivalence class $t_1$ alone, $t_2$ and $t_3$ and $t_4$, $t_5$ and $t_6$, $t_7$ and $t_8$ and lastly $t_9$ and $t_[10]$), $g(D) = \{19, 2[10], 36, 47, 58\}$ and $g(DL) = \{19, 2[10], 36, 47, 58\}$. With $\{19, 2[10], 36, 47, 58\}$ and $\{1, 2, 34, 56, 78, 9[10]\}$ as partitions, we have $f(\{19, 2[10], 36, 47, 58\}) = DL$ and $f(\{1, 2, 34, 56, 78, 9[10]\}) = R$.
	\end{example}
	
	\begin{proposition}
		The application couple $gc = (f, g)$ is a Galois connection between the power set lattice of the criterion set $\mathcal{C}$ and the partitions lattice of $Tid(r)$.
	\end{proposition}
	
	\begin{definition}[Agree concept closure operators]\index{Agree concept closure operators}
		The couple $gc = (f, g)$ being a particular case of the Galois connection, the compositions $f \circ g$  and $g \circ f$  of the two applications are closure operators (\cite{ganterFormalConceptAnalysis2024}) defined as follows:
		\[
		\begin{array}{lrcl}
			h: & \mathscr{P}(\mathcal{C})& \longrightarrow & \mathscr{P}(\mathcal{C}) \\
			& C & \longmapsto & \displaystyle f( g (C) ) = \bigcap_{\substack{C'\in Agree(r) \\
					C \subseteq C'}} C' \\
			h':& \Pi(Tid(r)) & \longrightarrow & \Pi(Tid(r)) \\
			& \pi & \longmapsto & \displaystyle g( f (\pi) ) = \bullet_{\substack{\pi'\in \Pi_{\mathscr{P}(\mathcal{C})} \\
					\pi \sqsubseteq \pi'}} \pi' \\
		\end{array}
		\]
	\end{definition}
	
	\begin{corollaire}
		$h$ and $h'$ satisfy the following properties:
		\begin{enumerate}
			\item $C \subseteq C' \Rightarrow h(C) \subseteq h(C')$ and  
			$\pi \sqsubseteq \pi' \Rightarrow h'(\pi) \sqsubseteq h'(\pi')$ (isotonicity)
			\item  $C \subseteq h(C)$ and $\pi \sqsubseteq h'(\pi)$ (extensivity)
			\item $h(C) = h(h(C))$ and $h'(\pi) = h'(h'(\pi))$ (idempotence)
		\end{enumerate}
	\end{corollaire}
	
	\begin{example}\label{ex:accords_fermeture}
		With the $\texttt{Pokémon}_{2}$ relation, taking the criterion sets $D$ and $DL$, based on the preceding example we have:
		\begin{itemize}
			\item $h(D) = f(g(D)) = f(\{19, 2[10], 36, 47, 58\}) = DL$  
			\item $h(DL) = f(g(DL)) = f(\{19, 2[10], 36, 47, 58\}) = DL$
		\end{itemize}
		With the partitions $\{19, 2[10], 36, 47, 58\}$ and $\{1, 2, 34, 56, 78, 9[10]\}$, we have:
		\begin{itemize}
			\item $h'(\{19, 2[10], 36, 47, 58\}) = g(f(\{19, 2[10], 36, 47, 58\})) = g(DL) = \{19, 2[10], 36, 47, 58\}$
			\item $h'(\{1, 2, 34, 56, 78, 9[10]\}) = g(f(\{1, 2, 34, 56, 78, 9[10]\})) = g(R) = \{1, 234, 56, 78, 9[10]\}$
		\end{itemize}
	\end{example}
	
	\begin{definition}[Agree concepts]\index{Agree concept}
		A relation $r$'s agree concept is a couple $(C, \pi)$ associating a criterion set to a identifier partition: $C \in \mathscr{P}(\mathcal{C})$ and $\pi \in \Pi(Tid(r))$. The elements of this couple shall be linked by the conditions $C = f(\pi)$ , $\pi = g(C) = \pi_C$.
		
		Let $c_a = (C_{c_a}, \pi_{c_a})$ be an agree concept of $r$, we call $\pi_{c_a}$ the extension of $c_a$ (noted $ext(c_a)$) and $C_{c_a}$ its intension (noted $int(c_a)$). The set of all agree concepts of a relation $r$ is noted $\texttt{AGREECONCEPTS}(r)$.
	\end{definition}
	
	\begin{theorem}[Agree concepts lattice]\index{Agree concepts lattice}
		Let $\texttt{AGREECONCEPTS}(r)$ be the relation $r$'s agree concepts set. $\langle \texttt{AGREECONCEPTS}(r), \leq \rangle$ is a ordered set forming a complete lattice named Agree concepts lattice \footnote{Let $(C_1, \pi_1)$, $(C_2, \pi_2) \in \texttt{AGREECONCEPTS}(r)$, $(C_1, \pi_1) \leq (C_2, \pi_2) \Leftrightarrow C_1 \subseteq C_2$ (or in other words $\pi_2 \sqsubseteq \pi_1$).}. We define, $\forall P \subseteq \texttt{AGREECONCEPTS}(r)$, the \emph{infimum} or lower bound ($\bigwedge$) and the \emph{supremum} or upper bound ($\bigvee$) as follows:
		\begin{align*}
			\bigwedge P &= (\bigcap_{c_a \in P} int(c_a),\ h'(+_{c_a \in P} ext(c_a))) \\
			\bigvee P &= (h(\bigcup_{c_a \in P} int(c_a)),\ \bullet_{c_a \in P} ext(c_a))
		\end{align*}
	\end{theorem}
	
	\begin{proof}
		As the couple $gc = (f, g)$ is a Galois connection, the Agree concepts lattice is a concept lattice according to the fundamental theorem of \cite{ganterFormalConceptAnalysis2024}.
	\end{proof}
	
	\begin{example}\label{ex:treillis_concepts_accords_des_novices}
		Figure~\ref{fig:treillis_des_concepts_accords_des_novices} shows the Hasse diagram of the relation $\texttt{Pokémon}_{1}$'s Agree concepts lattice. Assuming $c_a = (RD, \{1, 2, 3, 4, 5, 6, 78, 9, [10]\})$ and $c_b = (RL, \{1, 2, 34, 5, 6, 7, 8, 9, [10]\})$, we have:
		\begin{align*}
			c_a \wedge c_b = & (RD \cap RL, h'(\{19, 2[10], 36, 47, 58\} + \{1, 234, 56, 78, 9[10]\})) \\
			= & (R, h'(\{1, 234, 56, 78, 9[10]\})) = (R,\{1, 234, 56, 78, 9[10]\}) \\
			c_a \vee c_b   = & (h(RD \cup RL),\{19, 2[10], 36, 47, 58\} \bullet \{1, 234, 56, 78, 9[10]\}) \\
			= & (h(RDL), \{1, 2, 3, 4, 5, 6, 7, 8, 9, [10]\}) = (RDL, \{1, 2, 3, 4, 5, 6, 7, 8, 9, [10]\}) \\
		\end{align*}
	\end{example}
	
	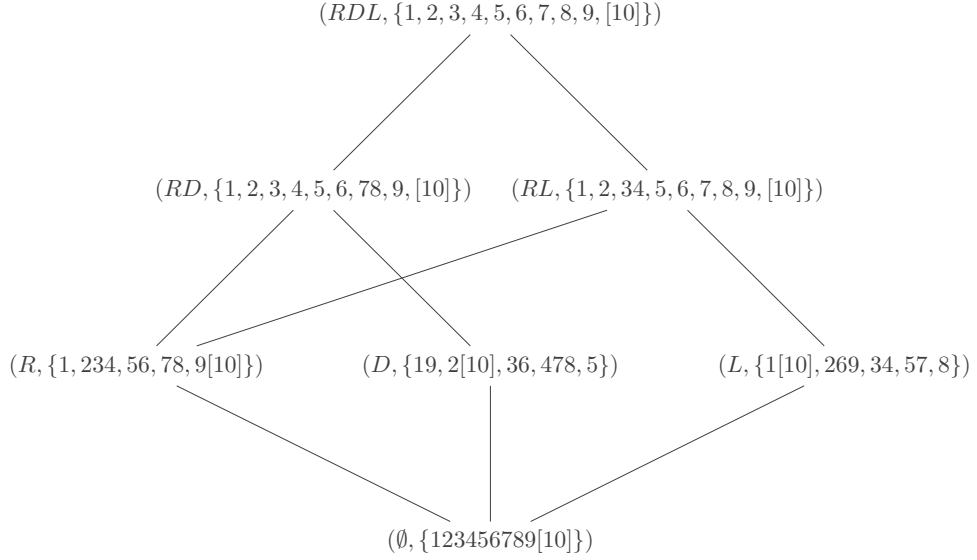
\begin{figure}
		\centering
		\resizebox{1\textwidth}{!}{
			\begin{tikzpicture}[
				line join=bevel,
				]
				
				\node (bottom) at (150pt, 0pt) {$(\emptyset, \{123456789[10]\})$};
				\node (p11) at (0pt, 75pt) {$(R, \{1, 234, 56, 78, 9[10]\})$};
				\node (p12) at (150pt, 75pt) {$(D, \{19, 2[10], 36, 478, 5\})$};
				\node (p13) at (300pt, 75pt) {$(L, \{1[10], 269, 34, 57, 8\})$};
				\node (p21) at (75pt, 150pt) {$(RD, \{1, 2, 3, 4, 5, 6, 78, 9, [10]\})$};
				\node (p22) at (225pt, 150pt) {$(RL, \{1, 2, 34, 5, 6, 7, 8, 9, [10]\})$};
				\node (top) at (150pt, 225pt) {$(RDL, \{1, 2, 3, 4, 5, 6, 7, 8, 9, [10]\})$};
				
				\draw [] (bottom) -- (p11);
				\draw [] (bottom) -- (p12);
				\draw [] (bottom) -- (p13);
				\draw [] (p11) -- (p21);
				\draw [] (p11) -- (p22);
				
				\draw [] (p12) -- (p21);
				
				\draw [] (p13) -- (p22);
				\draw [] (p21) -- (top);
				
				\draw [] (p22) -- (top);
			\end{tikzpicture}
		}
		\caption{Novice Agree concepts lattice Hasse diagram}\label{fig:treillis_des_concepts_accords_des_novices}
	\end{figure}
	
	\begin{example}\label{ex:treillis_concepts_accords_des_experts}
		Figure~\ref{fig:treillis_des_concepts_accords_des_experts} shows the Hasse diagram of the relation $\texttt{Pokémon}_{2}$'s Agree concepts lattice. The couple $(DL, \{19, 2[10], 36, 47, 58\})$ is an Agree concept since according to the examples~\ref{ex:accords_f_g} and~\ref{ex:accords_fermeture} we have $g(DL) = \{19, 2[10], 36, 47, 58\}$ and $f(\{19, 2[10], 36, 47, 58\}) = DL$. Whereas the couple $(D, \{19, 2[10], 36, 47, 58\})$ is not an Agree concept since $f(\{1, 2, 35, 4\}) \neq D$. Let $c_a = (DL, \{19, 2[10], 36, 47, 58\})$ and $c_b = (R, \{1, 234, 56, 78, 9[10]\})$ be two Agree concepts, we have: 
		\begin{align*}
			c_a \wedge c_b = & (DL \cap R, h'(\{19, 2[10], 36, 47, 58\} + \{1, 234, 56, 78, 9[10]\})) \\
			= & (\emptyset, h'(\{123456789[10]\})) = (\emptyset, \{123456789[10]\}) \\
			c_a \vee c_b   = & (h(DL \cup R), \{19, 2[10], 36, 47, 58\} \bullet \{1, 234, 56, 78, 9[10]\}) \\
			= & (h(RDL), \{1, 2, 3, 4, 5, 6, 7, 8, 9, [10]\}) = (RDL, \{1, 2, 3, 4, 5, 6, 7, 8, 9, [10]\}) \\
		\end{align*}
	\end{example}
	
	\begin{figure}
		\centering
		\resizebox{1\textwidth}{!}{
			\begin{tikzpicture}[
				line join=bevel,
				]
				
				\node (bottom) at (150pt, 0pt) {$(\emptyset, \{123456789[10]\})$};
				\node (p11) at (0pt, 75pt) {$(R, \{1, 234, 56, 78, 9[10]\})$};
				\node (p13) at (300pt, 75pt) {$(L, \{1369, 2[10], 47, 58\})$};
				\node (p21) at (225pt, 150pt) {$(DL, \{19, 2[10], 36, 47, 58\})$};
				\node (top) at (150pt, 225pt) {$(RDL, \{1, 2, 3, 4, 5, 6, 7, 8, 9, [10]\})$};
				
				\draw [] (bottom) -- (p11);
				\draw [] (bottom) -- (p21);
				\draw [] (bottom) -- (p13);
				\draw [] (p11) -- (top);
				
				\draw [] (p13) -- (p21);
				\draw [] (p21) -- (top);
			\end{tikzpicture}
		}
		\caption{Expert Agree concepts lattice Hasse diagram}\label{fig:treillis_des_concepts_accords_des_experts}
	\end{figure}
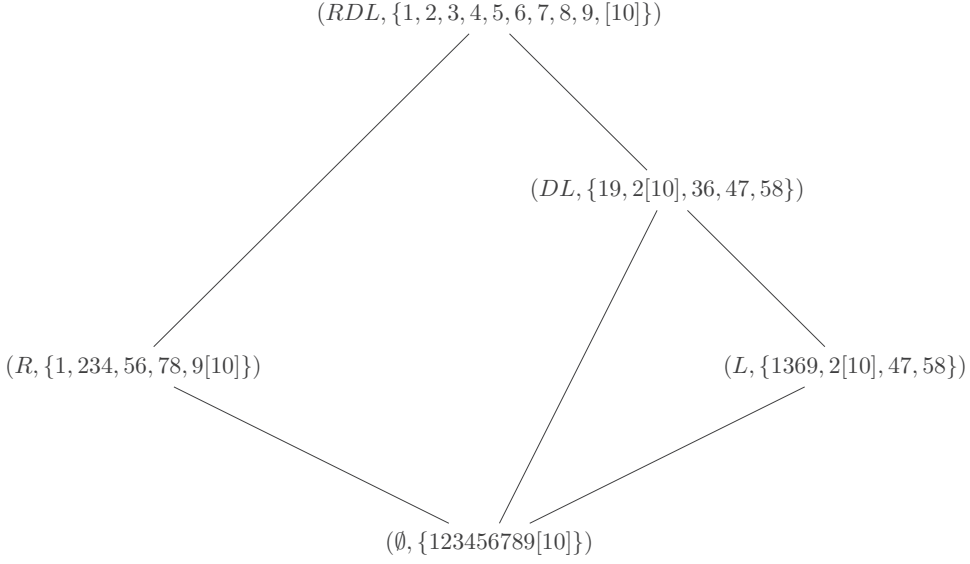
	
	\begin{proposition}\label{prop:partitions_egales}\index{Equivalent partitions}
		For any criterion set $C \subseteq \mathcal{C}$, the associated partition $\pi_C$ is identical to its closure's partition:
		\[\forall C \subseteq \mathcal{C}, \pi_C = \pi_{h(C)}\]
	\end{proposition}
	
	With the preceding proposition, we know that a criterion set $C$'s closure can be seen as the largest superset of $C$ with the same partition.
	
	\begin{proof}
		By definition $\forall C \subseteq \mathcal{C}$, $\pi_C = g(C)$ and $h(C) = f(g(C))$. Thus we have $\pi_{h(C)} = g(f(g(C))$. However, as the couple $gc = (f,\ g)$ is a Galois connection, we have $g \circ f \circ g  = g$ (\cite{ganterFormalConceptAnalysis2024}). Then, $\pi_{h(C)} = g(f(g(C))) = g(C) = \pi_C$.
	\end{proof}
	
	\begin{example}
		With the $\texttt{Pokémon}_{2}$ relation, assuming $D$ as criterion set, and according to examples~\ref{ex:accords_f_g} and~\ref{ex:accords_fermeture}, we have: $\pi_{D} = g(D) = \{19, 2[10], 36, 47, 58\}$ and $\pi_{h(D)} = g(h(D)) = g(RDL) = \{1, 2, 3, 4, 5, 6, 7, 8, 9, [10]\}$.
	\end{example}
	
	\subsection{Skyline concepts lattice and Skycube}
	
	The Skyline concepts lattice is a constrained Agree concepts lattice. We present a fundamental property to partially materialize a skycube by pruning some skycuboids. Then, we explain how such skycuboids can be easily rebuilt.
	
	\subsubsection{Skyline concepts lattice}
	
	Let $\pi_C$ be a partition of $r$ on $C$. By definition, $\forall t, [t]_C \in \pi_C$, all $[t]_C$ are indistinguishable on $C$ (all share the same projection on $C$). If $t$ is not dominated on $C$, all the tuples in its class will not be dominated (reciprocally dominated). Thus, dominance can be asserted for a single tuple of the class to know if the whole tuple set belongs in the skycuboid on $C$ or not. In order to optimise the skycuboid on $C$ computation from its partition $\pi_C$, we only keep a single tuple from each equivalence class. This set is noted $reps(\pi_C)$. This way, we reduce the input size by pruning all tuples leading to a lot of useless comparisons. To implement the specificities of the Skyline computation on a partition, we introduce the following operator.
	
	\begin{definition}[$\Pi\text{-}S$ operator]\index{$\Pi\text{-}S$ operator}
		Let $C$ a criterion set and $\pi$ a partition of $r$. We define new operator $\Pi\text{-}S$ as follow:
		\begin{align*}
			\Pi\text{-}S_C(\pi_C) &= \{[t_i] \in \pi_C \mid \forall t_j \in r \text{ we have } t_j \nsucc_C t_i \} \\
			&= \{[t_i] \in \pi_C \mid t_i \in \Pi\text{-}S_C(r)\}
		\end{align*}
	\end{definition}
	
	\begin{definition}[Skyline concept]\index{Skyline concept}
		Let $c_a = (C, \pi) \in \textsc{AGREECONCEPTS}(r)$ be a relation $r$'s Agree concept. The associated Skyline concept $c_s$ is defined as follows:
		\[c_s = (C, \Pi\text{-}S_{C}(\pi))\]
		
		The Skyline concepts and Agree concept counts are strictly equals. We note $\texttt{SKYLINECONCEPTS}(r)$ the set of skyline concepts associated to $r$'s Agree concepts.
	\end{definition}
	
	Skyline concepts are Agree concepts with constrained partitions. Thus, this type of concept is not necessarily ordered by the relation $\sqsubseteq$ between partitions. The order relation $\leq$ between Skylines concepts is defined as $\forall (C_1, \Pi\text{-}S_{C_1}(\pi_1))$, $(C_2, \Pi\text{-}S_{C_2}(\pi_2)) \in \texttt{SKYLINECONCEPTS}(r)$, thus $(C_1, \Pi\text{-}S_{C_1}(\pi_1)) \leq (C_2, \Pi\text{-}S_{C_2}(\pi_2)) \Leftrightarrow C_1 \subseteq C_2$.
	
	\begin{definition}[Skyline concepts lattice]\index{Skyline concepts lattice}
		$\langle \texttt{SKYLINECONCEPTS}(r), \leq\rangle$ is an ordered set forming a complete lattice named Skyline concepts lattice. It is isomorphic to the Agree concepts lattice.
	\end{definition}
	
	\begin{example}
		Figure~\ref{fig:treillis_des_concepts_skylines_des_novices} shows the $\texttt{Pokémon}_{1}$ relation's Skyline concepts lattice Hasse diagram. The $\texttt{Pokémon}_{1}$ relation Skycube lattice reduced representation, highlighting the evolution without the dimmed Skycuboid, is shown by figure~\ref{fig:skycube_treillis_reduit_3_1}.
	\end{example}

	\begin{figure}
		\centering
		\resizebox{1\textwidth}{!}{
			\begin{tikzpicture}[
				line join=bevel,
				]
				
				\node (bottom) at (150pt, 0pt) {$(\emptyset, \{123456789[10]\})$};
				\node (p11) at (0pt, 75pt) {$(R, \{56\})$};
				\node (p12) at (150pt, 75pt) {$(D, \{19\})$};
				\node (p13) at (300pt, 75pt) {$(L, \{1[10]\})$};
				\node (p21) at (75pt, 150pt) {$(RD, \{1, 6\})$};
				\node (p22) at (225pt, 150pt) {$(RL, \{1, 34, 6\})$};
				\node (top) at (150pt, 225pt) {$(RDL, \{1, 3, 6\})$};
				
				\draw [] (bottom) -- (p11);
				\draw [] (bottom) -- (p12);
				\draw [] (bottom) -- (p13);
				\draw [] (p11) -- (p21);
				\draw [] (p11) -- (p22);
				
				\draw [] (p12) -- (p21);
				
				\draw [] (p13) -- (p22);
				\draw [] (p21) -- (top);
				
				\draw [] (p22) -- (top);
			\end{tikzpicture}
		}
		\caption{Novice's Skyline concepts lattice Hasse diagram}\label{fig:treillis_des_concepts_skylines_des_novices}
	\end{figure}
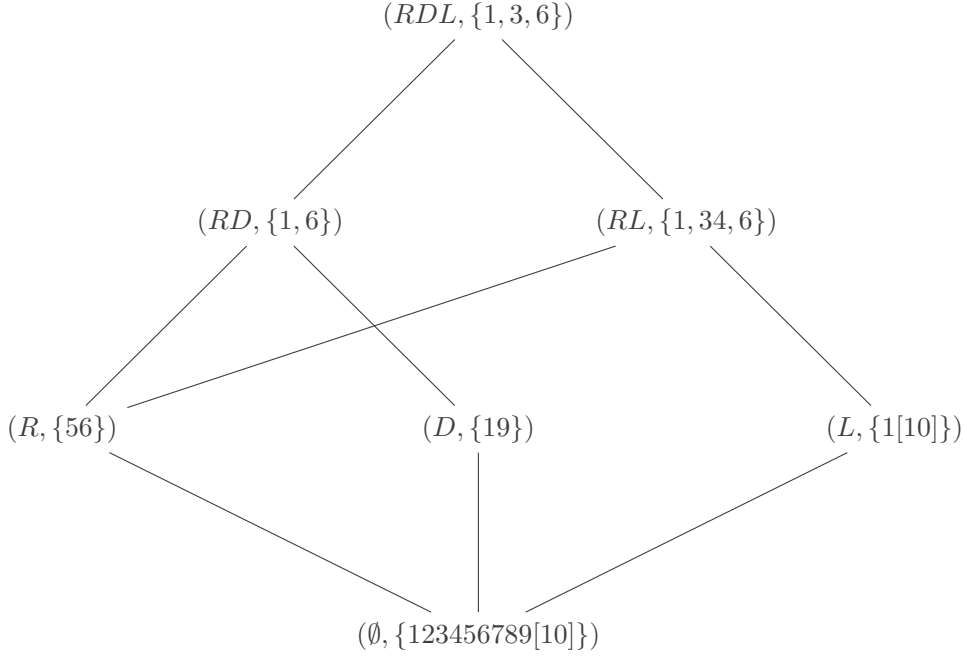
	
	\begin{figure}[htbp]
		\centering
		\tiny
		\resizebox{1\textwidth}{!}{
			\begin{tikzpicture}[
				line join=bevel
				]
				
				\node (bottom) at (150pt, 0pt) {$\emptyset$};
				\node (n11) at (50pt, 75pt)
				{
					\begin{tabular}{c|c}
						\toprule
						\texttt{Id} & \texttt{R} \\
						\midrule
						$5$ & $1$ \\
						$6$ & $1$ \\
						\bottomrule
					\end{tabular}
				};
				\node (n12) at (150pt, 75pt)
				{
					\begin{tabular}{c|c}
						\toprule
						\texttt{Id} & \texttt{D} \\
						\midrule
						$1$ & $25$ \\
						$9$ & $25$ \\
						\bottomrule
					\end{tabular}
				};
				\node (n13) at (250pt, 75pt)
				{
					\begin{tabular}{c|c}
						\toprule
						\texttt{Id} & \texttt{L} \\
						\midrule
						$1$ & $30$ \\
						$10$ & $30$ \\
						\bottomrule
					\end{tabular}
				};
				\node (n21) at (50pt, 150pt)
				{
					\begin{tabular}{c|cc}
						\toprule
						\texttt{Id} & \texttt{R} & \texttt{D} \\
						\midrule
						$1$ & $5$ & $25$ \\
						$6$ & $1$ & $35$ \\
						\bottomrule
					\end{tabular}
				};
				\node (n22) at (150pt, 150pt)
				{
					\begin{tabular}{c|cc}
						\toprule
						\texttt{Id} & \texttt{R} & \texttt{L} \\
						\midrule
						$1$ & $5$ & $30$ \\
						$3$ & $4$ & $40$ \\
						$4$ & $4$ & $40$ \\
						$6$ & $1$ & $50$ \\
						\bottomrule
					\end{tabular}
				};
				\node (n23) [lightgray] at (250pt, 150pt)
				{
					\begin{tabular}{c|cc}
						\toprule
						\texttt{Id} & \texttt{D} & \texttt{L} \\
						\midrule
						$1$ & $25$ & $30$ \\
						\bottomrule
					\end{tabular}
				};
				\node (top) at (150pt, 225pt)
				{
					\begin{tabular}{c|ccc}
						\toprule
						\texttt{Id} & \texttt{R} & \texttt{D} & \texttt{L} \\
						\midrule
						$1$ & $5$ & $25$ & $30$ \\
						$3$ & $4$ & $35$ & $40$ \\
						$6$ & $1$ & $35$ & $50$ \\
						\bottomrule
					\end{tabular}
				};
				
				\draw [stealth-] (bottom) -- (n11.south);
				\draw [stealth-] (bottom) -- (n12.south);
				\draw [stealth-] (bottom) -- (n13.south); 
				\draw [stealth-] (n11.90) -- (n21.270);
				\draw [stealth-] (n11.65) -- (n22.245);
				
				\draw [stealth-] (n12.115) -- (n21.290);
				\draw [stealth-, lightgray] (n12.65) -- (n23.245);
				
				\draw [stealth-] (n13.115) -- (n22.290);
				\draw [stealth-, lightgray] (n13.90) -- (n23.270);
				\draw [stealth-] (n21.north) -- (top);
				\draw [stealth-] (n22.north) -- (top);
				\draw [stealth-, lightgray] (n23.north) -- (top);
			\end{tikzpicture}
		}
		\caption{$\texttt{Pokémon}_{1}$ relation Skycube lattice reduced representation}\label{fig:skycube_treillis_reduit_3_1}
	\end{figure}
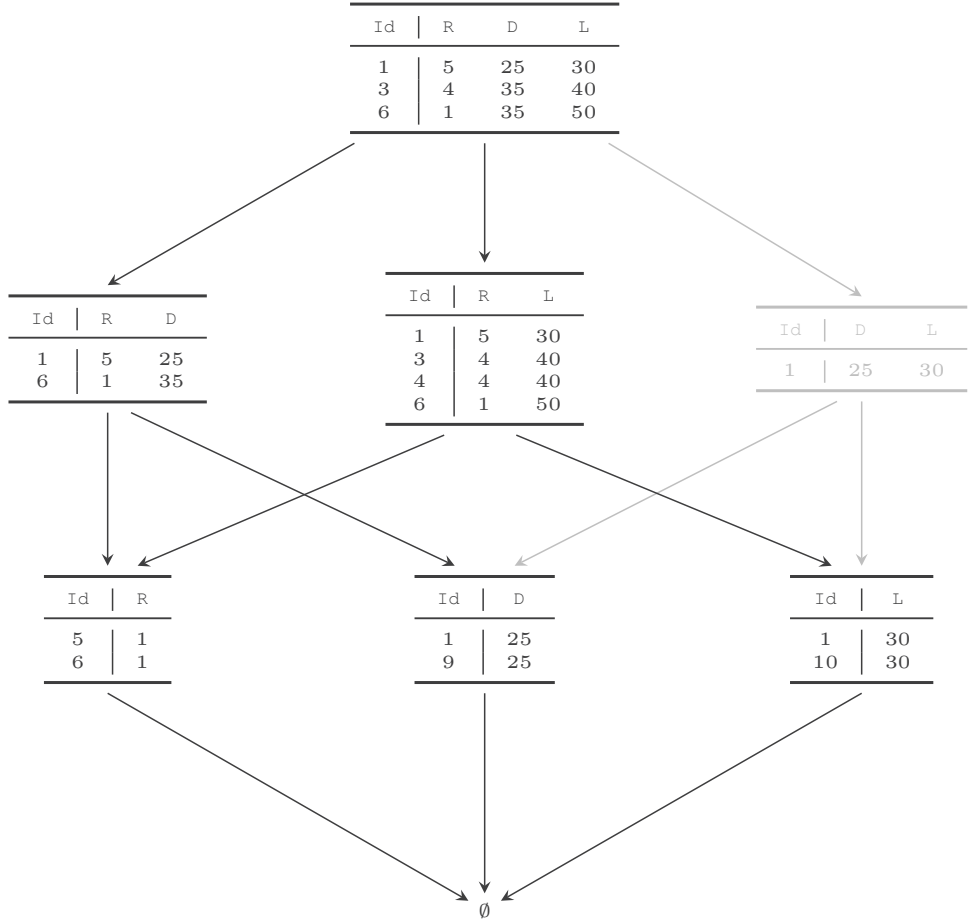
	
	\begin{example}
		Figure~\ref{fig:treillis_des_concepts_skylines_des_experts} shows the $\texttt{Pokémon}_{2}$ relation's Skyline concepts lattice Hasse diagram. According to the example~\ref{ex:treillis_concepts_accords_des_experts}, the couple $c_a = (DL, \{19, 2[10], 36, 47, 58\})$ is an agree concept. $c_s$, the associated Skyline concept, is $c_s = (DL, \Pi\text{-}S_{DL}(\{19, 2[10], 36, 47, 58\})) = (DL, \{19\})$. The other identifiers are pruned from the extension by $\Pi\text{-}S$ operator since corresponding tuples are dominated. $\texttt{Pokémon}_{2}$ relation's Skycube lattice reduced representation, highlighting the evolution without the dimmed Skycuboid, is shown by figure~\ref{fig:skycube_treillis_reduit_3_2}.
	\end{example}
	
	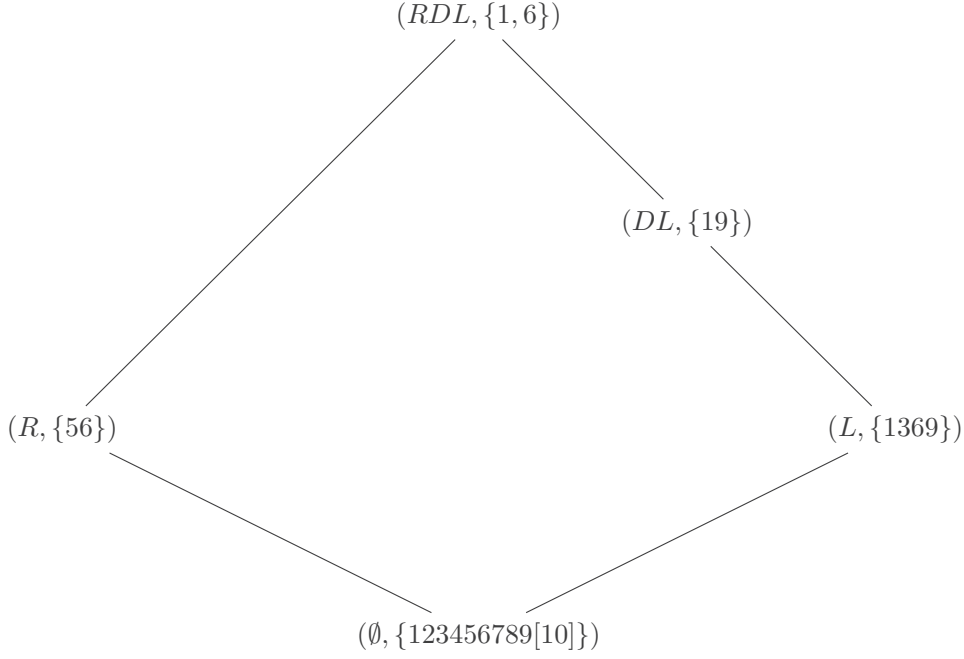
\begin{figure}
		\centering
		\resizebox{1\textwidth}{!}{
			\begin{tikzpicture}[
				line join=bevel,
				]
				
				\node (bottom) at (150pt, 0pt) {$(\emptyset, \{123456789[10]\})$};
				\node (p11) at (0pt, 75pt) {$(R, \{56\})$};
				\node (p13) at (300pt, 75pt) {$(L, \{1369\})$};
				\node (p21) at (225pt, 150pt) {$(DL, \{19\})$};
				\node (top) at (150pt, 225pt) {$(RDL, \{1, 6\})$};
				
				\draw [] (bottom) -- (p11);
				\draw [] (bottom) -- (p13);
				\draw [] (p11) -- (top);
				
				\draw [] (p13) -- (p21);
				\draw [] (p21) -- (top);
			\end{tikzpicture}
		}
		\caption{Expert's Skyline concepts lattice Hasse diagram}\label{fig:treillis_des_concepts_skylines_des_experts}
	\end{figure}
	
	\begin{figure}[htbp]
		\centering
		\tiny
		\resizebox{1\textwidth}{!}{
			\begin{tikzpicture}[
				line join=bevel
				]
				
				\node (bottom) at (150pt, 0pt) {$\emptyset$};
				\node (n11) at (50pt, 75pt)
				{
					\begin{tabular}{c|c}
						\toprule
						\texttt{Id} & \texttt{R} \\
						\midrule
						$5$ & $1$ \\
						$6$ & $1$ \\
						\bottomrule
					\end{tabular}
				};
				\node (n12) [lightgray] at (150pt, 75pt)
				{
					\begin{tabular}{c|c}
						\toprule
						\texttt{Id} & \texttt{D} \\
						\midrule
						$1$ & $20$ \\
						$9$ & $20$ \\
						\bottomrule
					\end{tabular}
				};
				\node (n13) at (250pt, 75pt)
				{
					\begin{tabular}{c|c}
						\toprule
						\texttt{Id} & \texttt{L} \\
						\midrule
						$1$ & $30$ \\
						$3$ & $30$ \\
						$6$ & $30$ \\
						$9$ & $30$ \\
						\bottomrule
					\end{tabular}
				};
				\node (n21) [lightgray] at (50pt, 150pt)
				{
					\begin{tabular}{c|cc}
						\toprule
						\texttt{Id} & \texttt{R} & \texttt{D} \\
						\midrule
						$1$ & $5$ & $20$ \\
						$6$ & $1$ & $30$ \\
						\bottomrule
					\end{tabular}
				};
				\node (n22) [lightgray] at (150pt, 150pt)
				{
					\begin{tabular}{c|cc}
						\toprule
						\texttt{Id} & \texttt{R} & \texttt{L} \\
						\midrule
						$6$ & $1$ & $30$ \\
						\bottomrule
					\end{tabular}
				};
				\node (n23) at (250pt, 150pt)
				{
					\begin{tabular}{c|cc}
						\toprule
						\texttt{Id} & \texttt{D} & \texttt{L} \\
						\midrule
						$1$ & $20$ & $30$ \\
						$9$ & $20$ & $30$ \\
						\bottomrule
					\end{tabular}
				};
				\node (top) at (150pt, 225pt)
				{
					\begin{tabular}{c|ccc}
						\toprule
						\texttt{Id} & \texttt{R} & \texttt{D} & \texttt{L} \\
						\midrule
						$1$ & $5$ & $20$ & $30$ \\
						$6$ & $3$ & $30$ & $30$ \\
						\bottomrule
					\end{tabular}
				};
				
				\draw [stealth-] (bottom) -- (n11.south);
				\draw [stealth-, lightgray] (bottom) -- (n12.south);
				\draw [stealth-] (bottom) -- (n13.south); 
				\draw [stealth-, lightgray] (n11.90) -- (n21.270);
				\draw [stealth-, lightgray] (n11.65) -- (n22.245);
				\draw [stealth-] (n11.77) -- (top);
				
				\draw [stealth-, lightgray] (n12.115) -- (n21.290);
				\draw [stealth-, lightgray] (n12.65) -- (n23.245);
				
				\draw [stealth-, lightgray] (n13.115) -- (n22.290);
				\draw [stealth-] (n13.90) -- (n23.280);
				\draw [stealth-, lightgray] (n21.north) -- (top);
				\draw [stealth-, lightgray] (n22.north) -- (top);
				\draw [stealth-] (n23.north) -- (top);
			\end{tikzpicture}
		}
		\caption{$\texttt{Pokémon}_{2}$ relation's Skycube lattice reduced representation}\label{fig:skycube_treillis_reduit_3_2}
	\end{figure}
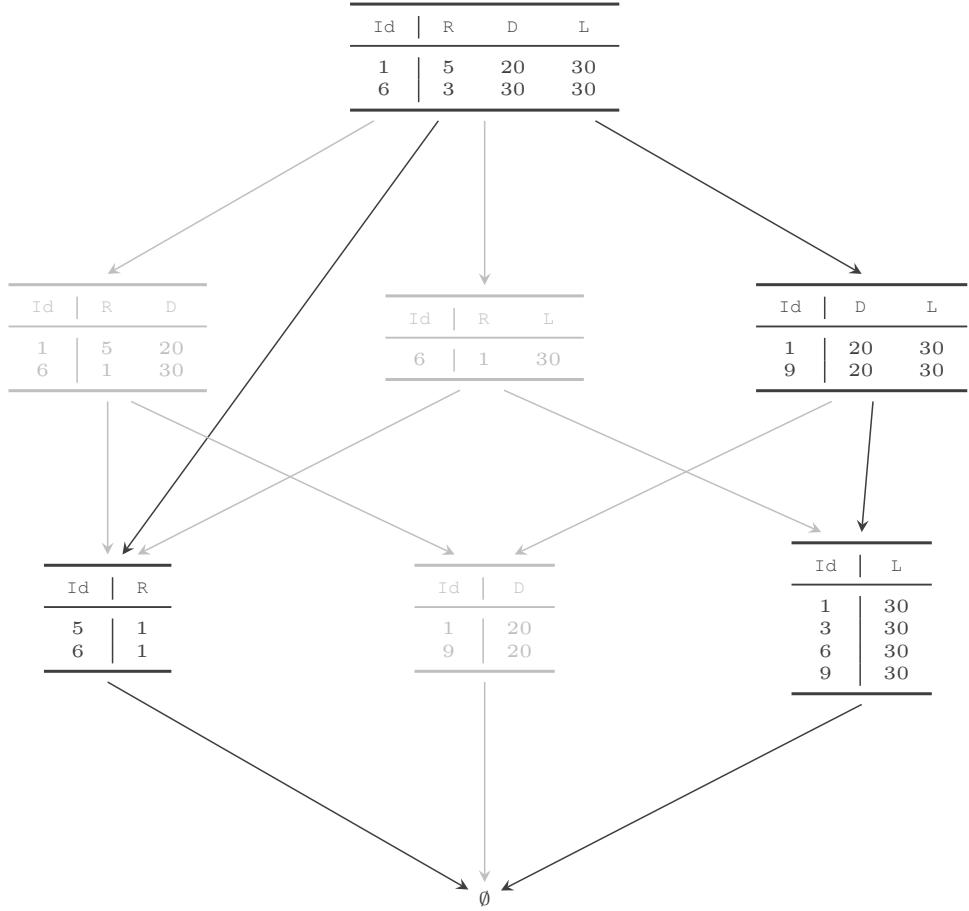
	
	\subsubsection{Skycube partial materialization}\label{ssec:materialisation_partielle_du_skycube}
	
	In this subsubsection, we propose the Skyline concept lattice as Skycubes' partial materialization. In order to obtain a reduced representation, the idea is to prune the most easily rebuildable Skycuboids.
	
	\begin{definition}[Disagree condition]\index{Disagree condition}
		Let $r$ be a relation and $C$ a criterion set. The disagree condition, or $DAC_C(r)$ is verified when:
		\[\nexists t_i,\ t_j \in r \text{ as } t_i[C] = t_j[C] \text{ with } i \neq j,\ C \subseteq \mathcal{C}\]
		
		When $DAC_C(r)$ is verified, $DAC_X(r)$ is also verified with $X$ being a superset of $C$.
	\end{definition}
	
	This disagree condition is a restricted version of the distinct value condition (\cite{yuanEfficientComputationSkyline2005}) since it is applied to projections and not to the individual values of each criterion.
	
	\begin{definition}[Dominance under disagree condition]\label{def:dominance_sous_condition_de_non_accord}\index{Dominance under disagree condition}
		Let $C \subseteq \mathcal{C}$ be a criterion set and $r$ a relation. The $t \succ_C t$ dominance relationship under $DAC_C(r)$ can be simplified as:
		\[\text{Let } C \text{ be as } DAC_C(r), \forall t_i, t_j \in r \text{ we have } t_j \succ_C t_i \Leftrightarrow \forall c \in C,\ t_j[c] \leq t_i[c] \text{ with } i \neq j\]
	\end{definition}
	
	\begin{lemma}\label{lem:skyline_dimension_superieure}
		Let $r$ be relation. For any $C \subseteq \mathcal{C}$ criterion set meeting $DAC_C(r)$, we have $S_C(r) \subseteq DAC_{C \cup c_0}(r)$ with $c_0 \in \mathcal{C} \smallsetminus C$.
	\end{lemma}
	
	\begin{proof}
		With $DAC_C(r)$, we have: $t_i \in S_C (r) \Rightarrow \forall t_j \in r$ with $i \neq j$, we have $t_j \nsucc_C t_i \Rightarrow \forall t_j \in r$ with $i \neq j$, $\exists c \in C, t_j[c] > t_i[c]$ (cf. definition~\ref{def:dominance_sous_condition_de_non_accord}) $\Rightarrow \forall t_j \in r$ with $i \neq j, \exists c \in C \cup c_0, t_j[c] > t_i[c] \Rightarrow$ $\forall t_j \in r$ with $i \neq j$, we have $t_j \nsucc_{C \cup c_0} t_i \Rightarrow$ $t_i \in S_{C \cup c_0} (r)$. So, with $DAC_C(r)$, we have $S_C (r) \subseteq S_{C \cup c_0} (r)$.
	\end{proof}
	
	The following counter example shows that the reciprocal property is not met.
	
	\begin{example}\label{ex:non_egalite_skylines}
		Let $r = \{ t_1, t_2\}$ be a relation (with $t_1=(0, 1)$ and $t_2=(1, 0)$) and $\mathcal{C} = \{ A, B\}$ the criterion set. $DAC_A(r)$ is met. We have $t_2 \notin S_A(r)$ since $t_1 \succ_A t_2$ while $t_2 \in S_{A \cup B}(r)$.
	\end{example}
	
	\begin{corollaire}
		Let $C$ be as $DAC_C(r)$ is met. By definition, for any $X$ ($C\subseteq X$) meeting $DAC_X(r)$, we have $S_C(r) \subseteq S_X(r)$.
	\end{corollaire}
	
	This property is met for any superset of $C$ up to the set of all criteria $\mathcal{C}$.
	
	\begin{theorem}[Fundamental theorem]\label{theoreme_fondamental}
		Let $r$ be a relation, $C$ a criterion set and $h(C)$ its closure. Then:
		\[\forall C \subseteq \mathcal{C}, S_C(r) \subseteq S_{h(C)}(r)\]
	\end{theorem}
	
	\begin{proof}
		By definition, $t_i \in S_C(r) \Leftrightarrow \exists [t] \in \Pi\text{-}S_C(\pi_C), i \in [t]$. We aim to prove $\Pi\text{-}S_C(\pi_C) \subseteq \Pi\text{-}S_C(\pi_{h(C)})$. We know that $\forall X, C \subseteq X \subseteq h(C), \pi_X = \pi_{C}$. Let $E = \{t_i \in r \mid i \in reps(\pi_{C})\}$ be this partition's representative tuples set. $E$ can only be used to compute Skyline on $X$ on the range $C \subseteq X \subseteq h(C)$, otherwise the equivalent classes no longer being equal, we can not ensure a total and proper Skyline building using equivalence classes elements. Furthermore, $DAC_X(E)$ is met since each representative tuple is distinct from another, so according to ~\ref{lem:skyline_dimension_superieure} lemma, we have $S_C(E) \subseteq S_{h(C)}(E)$ where $S_C(r) \subseteq S_{h(C)}(r)$.
	\end{proof}
	
	\begin{example}
		In the $\texttt{Pokémon}_{2}$ relation, we have $S_{RD}(\texttt{Pokémon}_{2}) = \{t_1, t_6\}$. Moreover, we have $h(RD) = RDL$ and $S_{RDL}(\texttt{Pokémon}_{2}) = \{t_1, t_6\}$. $S_{RD}(\texttt{Pokémon}_{2}) \subseteq S_{h(RD)}(\texttt{Pokémon}_{2})$ is therefore verified. Furthermore, we have $S_R(\texttt{Pokémon}_{2}) = \{t_5, t_6\}$ with $h(R) = R$. So, we note the non-inclusion $S_{R}(\texttt{Pokémon}_{2}) \nsubseteq S_{RD}(\texttt{Pokémon}_{2})$.
	\end{example}
	
	The preceding theorem shows an inclusion between some Skycuboids. More precisely, for any criterion set, there is an inclusion chain from its minimal generators up to its closure. This implies that any Skycuboid can be computed from another Skycuboid containing it. Thus, instead of using the whole relation, we only consider a restricted subset. The Skyline concepts are composed of the largest Skycuboids (depending on the inclusion) allowing to compute others. So, using only Skyline concepts (\emph{i.e.} non-closed Skycuboids pruning), we can easily rebuild missing Skycuboids by calculating the closure from which they can be computed.
	
	Furthermore, due to the Skyline concepts' equivalence classes, we prevent a lot of useless comparisons. Undistinctable tuples only being handled as groups, the computation complexity is no longer based on the tuple count but on the group (or equivalence class) count.
	
	So, the Skyline concepts lattice is a Skycube partial materialization, where non-materialized data can be efficiently computable.	
	
	\section{Emerging Skycube}\label{sec:calcul_du_skycube_emergent}\index{Emerging Skycube}
	
	As for the emerging datacube (\cite{nedjarEmergingCubesTrends2007}), to the best of our knowledge, there is no DBMS built-in solutions to compute an emerging Skycube. For the emerging Skycube to be a directly usable as an operator by a decision-maker, computation algorithms we present shall be integrable in a DBMS. With such a relational consideration, our proposition can be used in conjunction with existing ROLAP analysis tools. The emerging Skycube is therefore a specific datacube or even a specialized emerging datacube, and as with traditional datacube, query, exploration and search are allowed. The emerging datacube concept can be directly transposed as emerging Skycube, regardless of prior Skycube reduction by combining two unicompatible Skycubes, $r_1$ and $r_2$.
	
	An emerging tuple from $r_1$ to $r_2$ has a low measure in $r_1$ but a significative one in $r_2$. However a low measure can be either a low value or a high value, \emph{idem} for a significative measure. With these considerations, we redefine the emerging tuple concept.
	
	\begin{definition}[Emerging Tuple]\index{Emerging Tuple}
		A tuple $t\in CL(r_1 \cup r_2)$ is said emerging from $r_1$ to $r_2$ if and only if it satisfies two constraints, $C_1$ and $C_2$, as following:
		\[
		\left\lbrace
		\begin{array}{l}
			f_{val}(t, r_1) \geq MinThreshold_1 (C_1)  \\
			f_{val}(t, r_2) < MinThreshold_2 (C_2)
		\end{array}
		\right.
		\]
		
		Or:
		\[
		\left\lbrace
		\begin{array}{l}
			f_{val}(t, r_1) < MinThreshold_1 (C_1)  \\
			f_{val}(t, r_2) \geq MinThreshold_2 (C_2)
		\end{array}
		\right.
		\]
	\end{definition}
	
	In an emerging Skycube, the criterion notion, associated to a Skyline preference, coincide with the measure notion. In this document, we will use either depending on the context in order to preserve as much as possible traditional usage: in multi-criteria decision analysis context, we will prefer criterion, but in OLAP context, we will choose measure. Before this approach, emerging datacube computation was based on a single measure attribute in $r_1$, and its equivalent in $r_2$. A single measure now refers to a single criterion also associated to a single Skyline preference, which is quite irrelevant in a multi-criteria decision analysis. With Skycubes computation, and so with several Skyline preferences, including potentially heterogeneous ones, it becomes mandatory to use as many measure attributes in $r_1$ and the same count in $r_2$.
	
	Emerging Skycube computing and storage cost depends on the tuple count in all preserved Skycuboids, and so, also, on related criteria, Skyline preferences, and matching two-fold emergence constraints, still fixed by the decision maker, as for the emerging datacube. The emerging tuple count is bounded by the search space size, in other words by all the possible tuple set. The more emergence constraints are strong, the more the emerging Skycube is reduced. For constraints leading to an emerging Skycube of exploitable size, calibration is important, especially because Skycuboids' cardinality can be high, depending on criteria, even though it is significantly reduced by Skycubes' partial materialization previously seen. That's how measures' count and quality indirectly impact the emerging Skycube computation and storage cost, but also directly by adding dimensions, and corresponding tuples' values, to the input relation.
	
	The main difference with an emerging datacube, is that a Skycube's tuple is considered as emergent if and only if its emergence is observed for all Skycubes' measures. Therefore we introduce the measured emerging tuple.
	
	\begin{definition}[Measured emerging tuple]\label{def:tuple_de_skycube_emergent}\index{Measured emerging tuple}
		Let $r_1$ and $r_2$ be two unicompatible relations, $t \in CL(r_1 \cup r_2)$ a tuple and $f^{\mathcal{M}}$ an agregative function corresponding to the measure $\mathcal{M}$. $t$ is said measured emerging tuple if and only if it satisfies both constraints $C_1^{\mathcal{M}}$ and $C_2^{\mathcal{M}}$ as follows:
		\[
		\left\lbrace
		\begin{array}{l}
			f_{val}^{\mathcal{M}}(t, r_1) \geq MinThreshold_1 (C_1^{\mathcal{M}})  \\
			f_{val}^{\mathcal{M}}(t, r_2) < MinThreshold_2 (C_2^{\mathcal{M}})
		\end{array}
		\right.
		\]
		
		Or:
		\[
		\left\lbrace
		\begin{array}{l}
			f_{val}^{\mathcal{M}}(t, r_1) < MinThreshold_1 (C_1^{\mathcal{M}})  \\
			f_{val}^{\mathcal{M}}(t, r_2) \geq MinThreshold_2 (C_2^{\mathcal{M}})
		\end{array}
		\right.
		\]
	\end{definition}
	
	Traditionally, a measure with values remaining unchanged from a relation to another is irrelevant in an emerging datacube computation, and would be ignored and removed from input relations. However, with a emerging Skycube, we found that such a measure could be a useful criterion elsewhere, when computing Skycubes, and so has to be preserved. Therefore, we define the invariant measure.
	
	\begin{definition}[Invariant measure]\index{Invariant measure}
		Let $r_1$ and $r_2$ be two unicompatible relations and $f^{\mathcal{M}}$ an agregative function corresponding to the measure $\mathcal{M}$. $\mathcal{M}$ is called invariant measure if and only if: $\forall t \in CL(r_1 \cup r_2), f_{val}^{\mathcal{M}}(t, r_1) = f_{val}^{\mathcal{M}}(t, r_2)$.
	\end{definition}

	Before the emerging Skycube computation, Skycuboids are to be merged in a single relation. Indeed, in order to prevent computing all Skycuboids from both Skycubes, we merge both relations (\cite{diopCompositionMiningContexts2002}). This fusion's result is a new relation $r$ with the following schema: $\mathcal{R} = \mathcal{D} \cup \mathcal{M}_1 \cup \mathcal{M}_1' \cup \mathcal{M}_2 \cup \mathcal{M}_2' \dotsc \cup \mathcal{M}_n \cup \mathcal{M}_n'$ where $\mathcal{M}_1$, $\mathcal{M}_2$, \ldots, $\mathcal{M}_n$ are measures of $r_1$ and $\mathcal{M}_1'$, $\mathcal{M}_2'$, \ldots, $\mathcal{M}_n'$ measures of $r_2$. 
	
	Following the choice made for OLAP unpartitionned attributes, we propose to use special value $ALL$ for attributes excluded from Skycuboids. This value is a generalisation of all values of the attribute's domain (\cite{grayDataCubeRelational1996}). In other words, Skycuboids share a same schema and can be grouped in a single relation. Consequently, the fusion will generate lots of $ALL$ values in the relation, corresponding, for each tuple, to that many missing attributes from its initial Skycuboid. Contrary to the traditional emerging datacube, the deduplicating operation is already performed by Skycubes' computation. Thus, this step is skipped, and a merged relation is sufficient, a deduplicated merged relation (cf. table~\ref{tab:relation_fusionnee_dedoublonnee_pokemon}) becomes irrelevant.
	
	\begin{example}
		The merged relation is shown in table~\ref{tab:relation_fusionnee_pokemon} (with \texttt{R} for \texttt{Rarity}, \texttt{D}$_1$ for \texttt{Duration}$_1$, \texttt{L}$_1$ for \texttt{Loss}$_1$, \texttt{D}$_2$ for \texttt{Duration}$_2$ and \texttt{L}$_2$ for \texttt{Loss}$_2$). Separations in the table bound unicompatible Skycuboids tuples from both input Skycubes.
	\end{example}
	
	\begin{table}[htbp]
		\caption{Merged relation $\texttt{Pokémon}_{Mer}$}\label{tab:relation_fusionnee_pokemon}
		\centering
		\begin{minipage}{\linewidth}
				\begin{tabular}{c|cccc|ccccc} \toprule
					\texttt{RowId} & \texttt{Tier} & \texttt{Player} & \texttt{Opponent} & \ldots & \texttt{R}\footnote{Let $p$ be the percentage of drop for the Pokémon sequence (which is the multiplication of the percentage of drop for each Pokémon of the sequence), the \texttt{Rarity} score $r$ is calculated, on a scale of 0 to 10, as such: $\text{if } p = 1 \text{ then } r = 0 \text{, else } r = \lfloor max(\frac{(p - 1) \times 10 - (100 - e \times 0.9)}{e}, 0) \rfloor + 1 $.} & \texttt{Duration}\footnote{In total number of turns in the fight.} & \texttt{Loss}\footnote{In percentage.} & \texttt{D}$_2$\footnote{\emph{idem} \texttt{D}$_1$.} & \texttt{L}$_2$\footnote{\emph{idem} \texttt{L}$_1$.} \\
					\midrule
					$1$  & $UU$ & $A$ & $D$ & \ldots & $5$ & $25$ & $30$ & $20$ & $30$      \\
					$2$  & $OU$ & $B$ & $F$ & \ldots & $4$ & $35$ & $40$ & $ALL$ & $ALL$    \\
					$3$  & $OU$ & $C$ & $B$ & \ldots & $1$ & $35$ & $50$ & $30$ & $30$      \\
					\hline
					$4$  & $UU$ & $A$ & $D$ & \ldots & $5$ & $25$ & $ALL$ & $ALL$ & $ALL$   \\
					$5$  & $OU$ & $C$ & $B$ & \ldots & $1$ & $35$ & $ALL$ & $ALL$ & $ALL$   \\
					\hline
					$6$  & $UU$ & $A$ & $D$ & \ldots & $5$ & $ALL$ & $30$ & $ALL$ & $ALL$   \\
					$7$  & $OU$ & $B$ & $F$ & \ldots & $4$ & $ALL$ & $40$ & $ALL$ & $ALL$   \\
					$8$  & $OU$ & $B$ & $A$ & \ldots & $4$ & $ALL$ & $40$ & $ALL$ & $ALL$   \\
					$9$  & $OU$ & $C$ & $B$ & \ldots & $1$ & $ALL$ & $50$ & $ALL$ & $ALL$   \\
					\hline
					$10$ & $UU$ & $A$ & $D$ & \ldots & $ALL$ & $ALL$ & $ALL$ & $20$ & $30$  \\
					$11$ & $UU$ & $E$ & $D$ & \ldots & $ALL$ & $ALL$ & $ALL$ & $20$ & $30$  \\
					\hline
					$12$ & $OU$ & $C$ & $A$ & \ldots & $1$ & $ALL$ & $ALL$ & $ALL$ & $ALL$  \\
					$13$ & $OU$ & $C$ & $B$ & \ldots & $1$ & $ALL$ & $ALL$ & $ALL$ & $ALL$  \\	
					\hline
					$14$ & $UU$ & $A$ & $D$ & \ldots & $ALL$ & $25$ & $ALL$ & $ALL$ & $ALL$ \\
					$15$ & $UU$ & $E$ & $D$ & \ldots & $ALL$ & $25$ & $ALL$ & $ALL$ & $ALL$ \\
					\hline
					$16$ & $UU$ & $A$ & $D$ & \ldots & $ALL$ & $ALL$ & $30$ & $ALL$ & $30$  \\
					$17$ & $OU$ & $B$ & $F$ & \ldots & $ALL$ & $ALL$ & $ALL$ & $ALL$ & $30$ \\
					$18$ & $OU$ & $C$ & $B$ & \ldots & $ALL$ & $ALL$ & $ALL$ & $ALL$ & $30$ \\
					$19$ & $UU$ & $E$ & $D$ & \ldots & $ALL$ & $ALL$ & $ALL$ & $ALL$ & $30$ \\
					$20$ & $UU$ & $E$ & $E$ & \ldots & $ALL$ & $ALL$ & $30$ & $ALL$ & $ALL$  \\
					\bottomrule
				\end{tabular}
		\end{minipage}
	\end{table}
	
	To be able to clearly count a multiple measure trend reversal, we introduce the measured emergence rate computation.
	
	\begin{definition}[Measured emergence rate]\label{def:taux_d_emergence_mesure}\index{Measured emergence rate}
		Let $r_1$ and $r_2$ be two unicompatible relations, $t \in CL(r_1 \cup r_2)$ a tuple and $f^{\mathcal{M}}$ an agregative function corresponding to the measure $\mathcal{M}$. The measured emergence rate of $t$ from $r_1$ to $r_2$ in regards to the measure $\mathcal{M}$, noted $ER_{\mathcal{M}}(t)$, is defined by:
		\[
		ER_{\mathcal{M}_1}(t) =
		\left
		\lbrace
		\begin{array}{l}
			0 \text{ if } f_{val}^{\mathcal{M}}(t, r_1) = 0 \text{ and } f_{val}^{\mathcal{M}}(t, r_2) = 0 \\
			\infty \text{ if } f_{val}^{\mathcal{M}}(t, r_1) = 0 \text{ and } f_{val}^{\mathcal{M}}(t, r_2) \neq 0 \\
			\dfrac{f_{val}^{\mathcal{M}}(t, r_2)}{f_{val}^{\mathcal{M}}(t, r_1)} \text{ else.}
		\end{array}
		\right.
		\]
	\end{definition}
	
	An emerging Skycube has as many measured emergence rate as measures.
	
	\begin{example}
		Table~\ref{tab:tuples_emergents_mesures_de_la_relation_fusionnee_pokemon_selon_echec} shows measured emerging tuples of relation $\texttt{Pokémon}_{Mer}$ according to \texttt{Loss} with thresholds $MinThreshold_1 = 45$, to exceed, and $MinThreshold_2 = 45$, never to reach, and presented by Skycuboid. Measured emerging tuples, shown after the last table separation, correspond to the Skycuboid with only \texttt{Loss} criterion, and are dismissed from the emerging Skycube (cf. table~\ref{tab:skycube_emergent_de_la_relation_pokemon}) because these tuples are only composed by $ALL$ values for \texttt{Duration} measure.
	\end{example}
	
	\begin{table}[htbp]
		\caption{Relation $\texttt{Pokémon}_{Mer}$'s measured emerging tuples according to \texttt{Loss}}\label{tab:tuples_emergents_mesures_de_la_relation_fusionnee_pokemon_selon_echec}
		\centering
		\begin{tabular}{l|c}
			\toprule
			Measured emerging tuple & $ER_{\texttt{Loss}}$ \\
			\midrule
			$(ALL, ALL, B, \ldots)$ & $1.67$   \\
			$(ALL, C, ALL, \ldots)$ & $1.67$   \\
			$(OU, ALL, ALL, \ldots)$ & $3$   \\
			$(ALL, C, B, \ldots)$ & $1.67$   \\
			$(OU, ALL, B, \ldots)$ & $1.67$   \\
			$(OU, C, ALL, \ldots)$ & $1.67$   \\
			$(OU, C, B, \ldots)$   & $1.67$   \\
			\hline
			$(ALL, ALL, D, \ldots)$ & $\infty$ \\
			$(ALL, A, ALL, \ldots)$ & $\infty$ \\
			$(ALL, E, ALL, \ldots)$ & $\infty$ \\
			$(ALL, A, D, \ldots)$ & $\infty$ \\
			$(ALL, E, D, \ldots)$ & $\infty$ \\
			$(UU, ALL, ALL, \ldots)$ & $\infty$ \\
			$(UU, ALL, D, \ldots)$ & $\infty$ \\
			$(UU, A, ALL, \ldots)$ & $\infty$ \\
			$(UU, E, ALL, \ldots)$ & $\infty$ \\
			$(UU, A, D, \ldots)$ & $\infty$   \\
			$(UU, E, D, \ldots)$ & $\infty$   \\
			\hline
			$(ALL, ALL, F, \ldots)$ & $\infty$ \\
			$(ALL, ALL, B, \ldots)$ & $\infty$ \\
			$(ALL, B, ALL, \ldots)$ & $\infty$ \\
			$(ALL, C, ALL, \ldots)$ & $\infty$ \\
			$(OU, ALL, ALL, \ldots)$ & $\infty$ \\
			$(ALL, B, F, \ldots)$ & $\infty$   \\
			$(ALL, C, B, \ldots)$ & $\infty$   \\
			$(ALL, E, D, \ldots)$ & $\infty$   \\
			$(OU, ALL, F, \ldots)$ & $\infty$ \\
			$(OU, ALL, B, \ldots)$ & $\infty$ \\
			$(OU, B, ALL, \ldots)$ & $\infty$ \\
			$(OU, C, ALL, \ldots)$ & $\infty$ \\
			$(OU, B, F, \ldots)$ & $\infty$ \\
			$(OU, C, B, \ldots)$ & $\infty$ \\
			$(UU, E, D, \ldots)$ & $\infty$ \\
			\bottomrule
		\end{tabular}
	\end{table}
	
	From the merged relation, we obtain the abridged merged relation by:
	\begin{enumerate}
		\item deleting invariant measures, and corresponding tuple values;
		\item deleting tuples from Skycuboids not having at least all criteria from the second Skycube.
	\end{enumerate}
	
	Indeed, invariant measures are irrelevant for the emerging Skycube computation and can be pruned. Furthermore, tuples coming from Skycuboids with less criteria than its equivalent Skycuboid will never generate a trend reversal due to the lack of values to match the equivalent measures.
	
	\begin{example}
		The abridged relation of the merged relation (cf. table~\ref{tab:relation_fusionnee_pokemon}) is shown on the table~\ref{tab:relation_fusionnee_abregee_pokemon} (with \texttt{D}$_1$ for \texttt{Duration}$_1$, \texttt{L}$_1$ for \texttt{Loss}$_1$, \texttt{D}$_2$ for \texttt{Duration}$_2$ and \texttt{L}$_2$) for \texttt{Loss}$_2$). Separations presented in the merged relation's table are kept for practical reasons. This relation is used as input for IDEA platform.
	\end{example}
	
	\begin{table}[htbp]
		\caption{Abridged merged relation $\texttt{Pokémon}_{Mer}^+$}\label{tab:relation_fusionnee_abregee_pokemon}
		\centering
		\begin{minipage}{\linewidth}
			\centering
				\begin{tabular}{c|cccc|cccc} \toprule
					\texttt{RowId} & \texttt{Tier} & \texttt{Player} & \texttt{Opponent} & \ldots & \texttt{D}$_1$\footnote{In total number of turns in the fight.} & \texttt{L}$_1$\footnote{In percentage.} & \texttt{D}$_2$\footnote{\emph{idem} \texttt{D}$_1$.} & \texttt{L}$_2$\footnote{\emph{idem} \texttt{L}$_1$.} \\
					\midrule
					$1$  & $UU$& $A$ & $D$ & \ldots & $25$ & $30$ & $20$ & $30$     \\
					$2$  & $OU$ & $B$ & $F$ & \ldots & $35$ & $40$ & $ALL$ & $ALL$  \\
					$3$  & $OU$ & $C$ & $B$ & \ldots & $35$ & $50$ & $30$ & $30$    \\
					\hline
					$4$  & $UU$ & $A$ & $D$ & \ldots & $ALL$ & $ALL$ & $20$ & $30$  \\
					$5$  & $UU$ & $E$ & $D$ & \ldots & $ALL$ & $ALL$ & $20$ & $30$  \\
					\bottomrule
				\end{tabular}
		\end{minipage}
	\end{table}
	
	\begin{example}
		Table~\ref{tab:skycube_emergent_de_la_relation_pokemon} shows the emerging datacube of relation $\texttt{Pokémon}_{Mer}^+$ with following thresholds:
		\begin{itemize}
			\item \texttt{Duration} measure: $MinThreshold_1 = 35$ and $MinThreshold_2 = 35$;
			\item \texttt{Loss} measure: $MinThreshold_1 = 45$ and $MinThreshold_2 = 45$;
		\end{itemize}
		
		For this example, we select the \texttt{AVG} agregative function (\emph{\cite{peiPushingConvertibleConstraints2004}}). Separations presented in the merged relation's table are kept for practical reasons.
	\end{example}
	
	\begin{table}[htbp]
		\caption{Emerging Skycube of relation $\texttt{Pokémon}_{Mer}^+$}\label{tab:skycube_emergent_de_la_relation_pokemon}
		\centering
		\begin{tabular}{l|cc}
			\toprule
			Emerging tuple & $ER_{\texttt{Duration}}$ & $ER_{\texttt{Loss}}$ \\
			\midrule
			$(ALL, ALL, B, \ldots)$ & $1.17$ & $1.67$      \\
			$(ALL, C, ALL, \ldots)$ & $1.17$ & $1.67$      \\
			$(OU, ALL, ALL, \ldots)$ & $2.33$ & $3$        \\
			$(ALL, C, B, \ldots)$ & $1.17$ & $1.67$        \\
			$(OU, ALL, B, \ldots)$ & $1.17$ & $1.67$       \\
			$(OU, C, ALL, \ldots)$ & $1.17$ & $1.67$       \\
			$(OU, C, B, \ldots)$ & $1.17$ & $1.67$         \\
			\hline
			$(ALL, ALL, F, \ldots)$ & $\infty$ & $\infty$  \\
			$(ALL, ALL, B, \ldots)$ & $\infty$ & $\infty$  \\
			$(ALL, B, ALL, \ldots)$ & $\infty$ & $\infty$  \\
			$(ALL, C, ALL, \ldots)$ & $\infty$ & $\infty$  \\
			$(OU, ALL, ALL, \ldots)$ & $\infty$ & $\infty$ \\
			$(ALL, B, F, \ldots)$ & $\infty$ & $\infty$    \\
			$(ALL, C, B, \ldots)$ & $\infty$ & $\infty$    \\
			$(ALL, E, D, \ldots)$ & $\infty$ & $\infty$    \\
			$(OU, ALL, F, \ldots)$ & $\infty$ & $\infty$   \\
			$(OU, ALL, B, \ldots)$ & $\infty$ & $\infty$   \\
			$(OU, B, ALL, \ldots)$ & $\infty$ & $\infty$   \\
			$(OU, C, ALL, \ldots)$ & $\infty$ & $\infty$   \\
			$(OU, B, F, \ldots)$ & $\infty$ & $\infty$     \\
			$(OU, C, B, \ldots)$ & $\infty$ & $\infty$     \\
			$(UU, E, D, \ldots)$ & $\infty$ & $\infty$     \\
			\bottomrule
		\end{tabular}
	\end{table}
	
	The emerging Skycube is computable by E-IDEA algorithm (\cite{nedjarEmergingDataCube2013, nedjarExtractingSemanticsOLAP2011}), as long as the input format is respected, like any emerging datacube, with the particularity that, as of now, the algorithm shall be called on each measure.
	
	\section{Emerging closed Skycube}\index{Emerging closed L-Skycube}\index{Emerging closed Skycube}
	
	An emerging datacube's reduced representation's main interest, especially with the emerging closed datacube, can be directly transposed to the emerging closed Skycube.
	
	Despite being significantly useful and offering a vastly reduced representation  (\cite{nedjarEmergingCubesBorders2009}), L (for lower) borders and U (for upper) borders generates a data loss. Indeed, measure values, since aggregated, can no longer be recovered. If the decision maker find this loss acceptable, then, use of borders as a representation is quite appropriated, alongside F-IDEA algorithm.
	
	\begin{example}
		With example abridged merged relation $\texttt{Pokémon}_{Mer}^+$, table~\ref{tab:bordures_l_du_skycube_emergent_de_la_relation_pokemon} shows the $L$ border for the emerging Skycube, keeping a similar value for thresholds (for \texttt{Duration} $MinThreshold_1 = 35, MinThreshold_2 = 35$ and for \texttt{Loss} $MinThreshold_1 = 45, MinThreshold_2 = 45$). Separations presented in the merged relation's table $\texttt{Pokémon}_{Mer}^+$ are kept for practical reasons. 
	\end{example}
	
	\begin{table}[htbp]
		\caption{ $\texttt{Pokémon}_{Mer}^+$ relation's emerging Skycube's $L$ border}\label{tab:bordures_l_du_skycube_emergent_de_la_relation_pokemon}
		\centering
		\begin{tabular}{c|l}
			\toprule
			$L$ & $(ALL, ALL, B, \ldots)$ \\
			& $(ALL, C, ALL, \ldots)$ \\
			& $(OU, ALL, ALL, \ldots)$ \\
			\hline
			& $(ALL, ALL, F, \ldots)$ \\
			& $(ALL, ALL, B, \ldots)$ \\
			& $(ALL, B, ALL, \ldots)$ \\
			& $(ALL, C, ALL, \ldots)$ \\
			& $(OU, ALL, ALL, \ldots)$ \\
			& $(ALL, E, D, \ldots)$ \\
			\bottomrule
		\end{tabular}
	\end{table}
	
	Nevertheless, to handle any OLAP query, and so recover all measures' values, the emerging closed Skycube's reduced representation is best suited. As stated before, this concept is directly transposed from the closed emerging datacube and, as long as the input format is respected, IDEA algorithms can be applied as well. Precisely, this representation encompass the set of closed emerging Skycube's tuples, to which can be added any information required to assure a lossless representation, reproduced by the L border, explaining the closed emerging L-datacube transposition into closed emerging L-Skycube. In this case, $\mathbb{C}$-IDEA algorithm is applied to compute this reduced representation.
	
	\begin{example}
		With example abridged merged relation $\texttt{Pokémon}_{Mer}^+$, table~\ref{tab:tuples_fermes_emergents_de_la_relation_pokemon} shows closed emerging tuples for the emerging Skycube, keeping a similar value for thresholds (for \texttt{Duration} $MinThreshold_1 = 35, MinThreshold_2 = 35$ and for \texttt{Loss} $MinThreshold_1 = 45, MinThreshold_2 = 45$). Separations presented in the merged relation's table $\texttt{Pokémon}_{Mer}^+$ are kept for practical reasons. 
	\end{example}
	
	\begin{table}[htbp]
		\caption{$\texttt{Pokémon}_{Mer}^+$ relation's closed emerging tuples}\label{tab:tuples_fermes_emergents_de_la_relation_pokemon}
		\centering
		\begin{tabular}{l|cc}
			\toprule 
			Closed emerging tuple & $ER_{\texttt{Duration}}$ & $ER_{\texttt{Loss}}$ \\
			\midrule 
			$(OU, ALL, ALL, \ldots)$ & $2.33$ & $3$        \\
			$(OU, C, B, \ldots)$ & $1.17$ & $1.67$         \\
			\hline
			$(OU, ALL, ALL, \ldots)$ & $\infty$ & $\infty$ \\
			$(OU, B, F, \ldots)$ & $\infty$ & $\infty$     \\
			$(OU, C, B, \ldots)$ & $\infty$ & $\infty$     \\
			$(UU, E, D, \ldots)$ & $\infty$ & $\infty$     \\
			\bottomrule
		\end{tabular}
	\end{table}
	
	\begin{example}
		With example abridged merged relation $\texttt{Pokémon}_{Mer}^+$, table~\ref{tab:l_skycube_ferme_emergent} shows closed emerging L-Skycube, keeping a similar value for thresholds (for \texttt{Duration} $MinThreshold_1 = 35, MinThreshold_2 = 35$ and for \texttt{Loss} $MinThreshold_1 = 45, MinThreshold_2 = 45$). Separations presented in the merged relation's table $\texttt{Pokémon}_{Mer}^+$ are kept for practical reasons. 
	\end{example}
	
	\begin{table}[htbp]
		\caption{Closed emerging L-Skycube}\label{tab:l_skycube_ferme_emergent}
		\centering
		\begin{minipage}{\linewidth}
			\centering
			\begin{tabular}{l|cc}
				\toprule
				Closed emerging tuple\footnote{With: \begin{itemize}\item A : 121, 113, 006 (Starmie, Chansey, Charizard); \item B : 065, 103, 065 (Alakazam, Exeggutor, Alakazam); \item C : 121, 113, 080 (Starmie, Chansey, Slowbro); \item D : 065, 113, 143 (Alakazam, Chansey, Snorlax); \item E : 065, 040, 065 (Alakazam, Wigglytuff, Alakazam); \item F : 121, 113, 121 (Starmie, Chansey, Starmie).\end{itemize}} & $ER_{\texttt{Duration}}$ & $ER_{\texttt{Loss}}$ \\
				\midrule
				$(ALL, ALL, B, \ldots)$ & $1.17$ & $1.67$      \\
				$(ALL, C, ALL, \ldots)$ & $1.17$ & $1.67$      \\
				$(OU, ALL, ALL, \ldots)$ & $2.33$ & $3$        \\
				$(OU, C, B, \ldots)$ & $1.17$ & $1.67$         \\
				\hline
				$(ALL, ALL, F, \ldots)$ & $\infty$ & $\infty$  \\
				$(ALL, ALL, B, \ldots)$ & $\infty$ & $\infty$  \\
				$(ALL, B, ALL, \ldots)$ & $\infty$ & $\infty$  \\
				$(ALL, C, ALL, \ldots)$ & $\infty$ & $\infty$  \\
				$(OU, ALL, ALL, \ldots)$ & $\infty$ & $\infty$ \\
				$(ALL, E, D, \ldots)$ & $\infty$ & $\infty$    \\
				$(OU, B, F, \ldots)$ & $\infty$ & $\infty$     \\
				$(OU, C, B, \ldots)$ & $\infty$ & $\infty$     \\
				$(UU, E, D, \ldots)$ & $\infty$ & $\infty$     \\
				\bottomrule
			\end{tabular}
		\end{minipage}
	\end{table}
	
	Each emerging Skycube's reduced representation is even more reduced due to the Skycube partial materialization. As the Skycube computation drastically lowers the tuple count from the initial relation, further reduced representation are less effective in an emerging Skycube than in an emerging datacube.
	
	\section{Experimental evaluation}

	Compute a Emerging Skycube is a variant of a Skycube computation, and the most significant part is due to partial materialization.

	In this section, we present the validation of the Skycube partial materialization theoretical approach through performance evaluations on generated synthetic universities' databases of dimension ten and with a million tuples. We have generated realistic and representative sets of decorrelated synthetic data. These evaluations are directed towards a two-fold objective: measure the rediction in size and evaluate Skyline queries' response time.
	
	The experimental evaluation was performed on an Intel(R) Xeon(R) W-11955M CPU @ 2.60GHz 2.61 GHz, with 32Gb RAM, powered by Linux. The source code was written in Python 3.8 and interpreted with PyPy 3.9. PyPy is an alternative implementation of the Python programming language, designed to be faster and more efficient in terms of memory consumption than the standard Python implementation CPython. On average, PyPy 3.9 is 4.8 times faster than CPython 3.7. Times shown are in seconds, measured as processor processing time and assuming a default value of 8~ms per page default.
	
	\subsubsection{Size reduction}
	
	\begin{figure}[htbp]
		\centering
		\resizebox{1\textwidth}{!}{
			\begin{tikzpicture}[
				line join=bevel,
				bigrednode/.style={shape=circle, fill=brightmaroon, draw=black, line width=1pt},
				bigbluenode/.style={shape=circle, fill=skyblue, draw=black, line width=1pt}
				]
				
				\draw[-stealth] (0pt, 0pt) -- (300pt, 0pt) node[anchor=north west] {Cardinality};
				\draw[-stealth] (0pt, 0pt) -- (0pt, 300pt) node[anchor=south] {Tuples count};
				
				\foreach \y/\ytext in {0pt/$0$, 35pt/$1000000$, 70pt/$2000000$, 105pt/$3000000$, 140pt/$4000000$, 175pt/$5000000$, 210pt/$6000000$, 245pt/$7000000$, 280pt/$8000000$} {
					\draw (2pt, \y) -- (-2pt, \y) node[left] {$\ytext\strut$};
				}
				\foreach \x/\xtext in {0pt/$100$, 140pt/$1000$, 280pt/$10000$} {
					\draw (\x, 2pt) -- (\x, -2pt) node[below] {$\xtext\strut$};
				}
				
				\draw[skyblue, line width=2pt] (0pt, 85pt) -- (42pt, 75pt) -- (68pt, 63pt) -- (84pt, 60pt) -- (98pt, 57pt) -- (110pt, 57pt) -- (119pt, 54pt) -- (126pt, 51pt) -- (133pt, 48pt) -- (140pt, 45pt) -- (182pt, 33pt) -- (208pt, 30pt) -- (224pt, 30pt) -- (238pt, 29pt) -- (250pt, 28pt) -- (259pt, 27pt) -- (268pt, 26pt) -- (275pt, 25pt);
				\draw[brightmaroon, line width=2pt] (0pt, 131pt) -- (42pt, 166pt) -- (68pt, 181pt) -- (84pt, 193pt) -- (98pt, 199pt) -- (110pt, 205pt) -- (119pt, 208pt) -- (126pt, 211pt) -- (133pt, 217pt) -- (140pt, 220pt) -- (182pt, 232pt) -- (208pt, 235pt) -- (224pt, 241pt) -- (238pt, 242pt) -- (250pt, 244pt) -- (259pt, 245pt) -- (268pt, 247pt) -- (275pt, 250pt);
				
				\filldraw[color=black, fill=skyblue] (0pt, 85pt) circle (2pt);
				\filldraw[color=black, fill=skyblue] (42pt, 75pt) circle (2pt);
				\filldraw[color=black, fill=skyblue] (68pt, 63pt) circle (2pt);
				\filldraw[color=black, fill=skyblue] (84pt, 60pt) circle (2pt);
				\filldraw[color=black, fill=skyblue] (98pt, 57pt) circle (2pt);
				\filldraw[color=black, fill=skyblue] (110pt, 57pt) circle (2pt);
				\filldraw[color=black, fill=skyblue] (119pt, 54pt) circle (2pt);
				\filldraw[color=black, fill=skyblue] (126pt, 51pt) circle (2pt);
				\filldraw[color=black, fill=skyblue] (133pt, 48pt) circle (2pt);
				\filldraw[color=black, fill=skyblue] (140pt, 45pt) circle (2pt);
				\filldraw[color=black, fill=skyblue] (182pt, 33pt) circle (2pt);
				\filldraw[color=black, fill=skyblue] (208pt, 30pt) circle (2pt);
				\filldraw[color=black, fill=skyblue] (224pt, 30pt) circle (2pt);
				\filldraw[color=black, fill=skyblue] (238pt, 29pt) circle (2pt);
				\filldraw[color=black, fill=skyblue] (250pt, 28pt) circle (2pt);
				\filldraw[color=black, fill=skyblue] (259pt, 27pt) circle (2pt);
				\filldraw[color=black, fill=skyblue] (268pt, 26pt) circle (2pt);
				\filldraw[color=black, fill=skyblue] (275pt, 25pt) circle (2pt);
				
				\filldraw[color=black, fill=brightmaroon] (0pt, 131pt) circle (2pt);
				\filldraw[color=black, fill=brightmaroon] (42pt, 166pt) circle (2pt);
				\filldraw[color=black, fill=brightmaroon] (68pt, 181pt) circle (2pt);
				\filldraw[color=black, fill=brightmaroon] (84pt, 193pt) circle (2pt);
				\filldraw[color=black, fill=brightmaroon] (98pt, 199pt) circle (2pt);
				\filldraw[color=black, fill=brightmaroon] (110pt, 205pt) circle (2pt);
				\filldraw[color=black, fill=brightmaroon] (119pt, 208pt) circle (2pt);
				\filldraw[color=black, fill=brightmaroon] (126pt, 211pt) circle (2pt);
				\filldraw[color=black, fill=brightmaroon] (133pt, 217pt) circle (2pt);
				\filldraw[color=black, fill=brightmaroon] (140pt, 220pt) circle (2pt);
				\filldraw[color=black, fill=brightmaroon] (182pt, 232pt) circle (2pt);
				\filldraw[color=black, fill=brightmaroon] (208pt, 235pt) circle (2pt);
				\filldraw[color=black, fill=brightmaroon] (224pt, 241pt) circle (2pt);
				\filldraw[color=black, fill=brightmaroon] (238pt, 242pt) circle (2pt);
				\filldraw[color=black, fill=brightmaroon] (250pt, 244pt) circle (2pt);
				\filldraw[color=black, fill=brightmaroon] (259pt, 245pt) circle (2pt);
				\filldraw[color=black, fill=brightmaroon] (268pt, 247pt) circle (2pt);
				\filldraw[color=black, fill=brightmaroon] (275pt, 250pt) circle (2pt);
				
				\matrix [below left] at (current bounding box.north east) {
					\node [bigrednode, label=right:Skycube] {}; \\
					\node [bigbluenode, label=right:Skyline concept lattice] {}; \\
				};
			\end{tikzpicture}
		}
		\caption{Size reduction}\label{fig:reduction_de_la_taille}
	\end{figure}
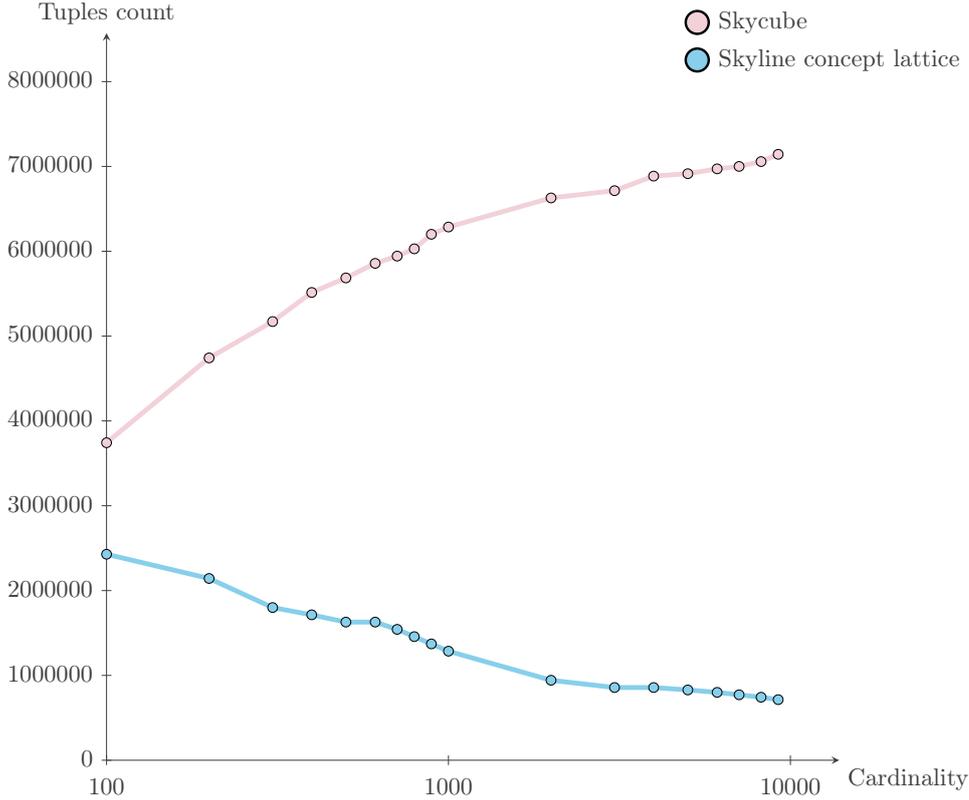
	
	In this experimentation, we measure at the same time both the Skycube's size and the Skyline concept lattice' size to compare them. Figure~\ref{fig:reduction_de_la_taille} shows the results obtained when criteria cardinality varies from 100 to 10000. As the cardinality increases, the Skycube becomes way bigger while its representation materialization's size rises marginally.
	
	\subsubsection{Queries' response time}
	
	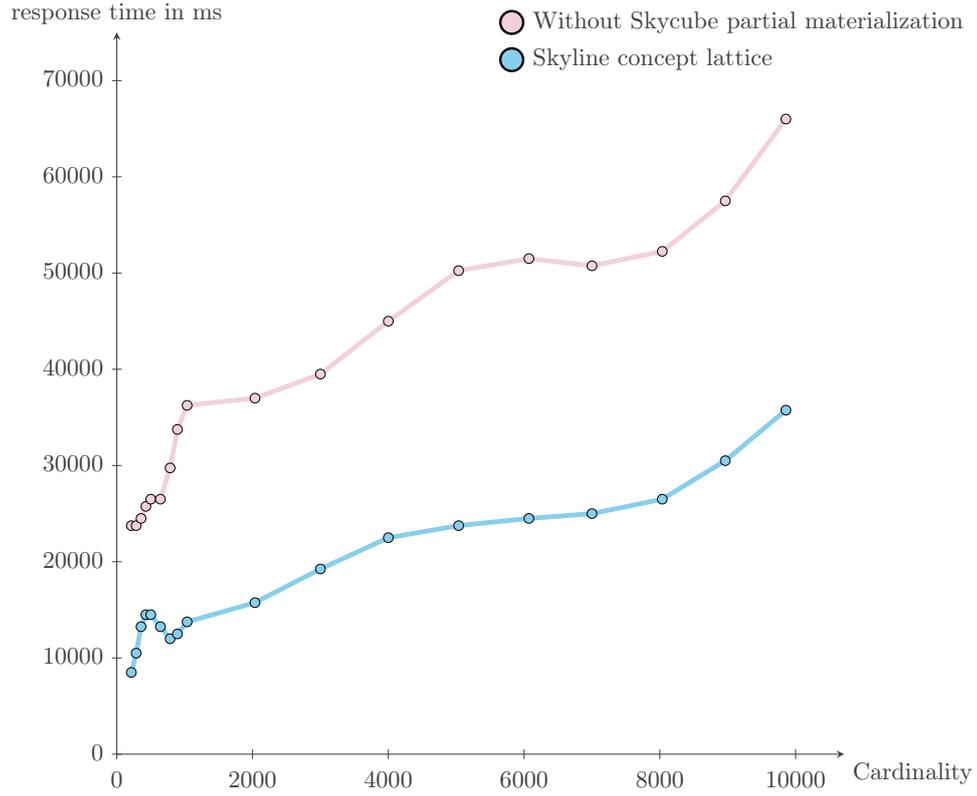
\begin{figure}[htbp]
		\centering
		\resizebox{1\textwidth}{!}{
			\begin{tikzpicture}[
				line join=bevel,
				bigrednode/.style={shape=circle, fill=brightmaroon, draw=black, line width=1pt},
				bigbluenode/.style={shape=circle, fill=skyblue, draw=black, line width=1pt}
				]
				
				\draw[-stealth] (0pt, 0pt) -- (300pt, 0pt) node[anchor=north west] {Cardinality};
				\draw[-stealth] (0pt, 0pt) -- (0pt, 300pt) node[anchor=south] {response time in ms};
				
				\foreach \y/\ytext in {0pt/$0$, 40pt/$10000$, 80pt/$20000$, 120pt/$30000$, 160pt/$40000$, 200pt/$50000$, 240pt/$60000$, 280pt/$70000$} {
					\draw (2pt, \y) -- (-2pt, \y) node[left] {$\ytext\strut$};
				}
				\foreach \x/\xtext in {0pt/$0$, 56pt/$2000$, 112pt/$4000$, 168pt/$6000$, 224pt/$8000$, 280pt/$10000$} {
					\draw (\x, 2pt) -- (\x, -2pt) node[below] {$\xtext\strut$};
				}
				
				\draw[skyblue, line width=2pt] (6pt, 34pt) -- (8pt, 42pt) -- (10pt, 53pt) -- (12pt, 58pt) -- (14pt, 58pt) -- (18pt, 53pt) -- (22pt, 48pt) -- (25pt, 50pt) -- (29pt, 55pt) -- (57pt, 63pt) -- (84pt, 77pt) -- (112pt, 90pt) -- (141pt, 95pt) -- (170pt, 98pt) -- (196pt, 100pt) -- (225pt, 106pt) -- (251pt, 122pt) -- (276pt, 143pt);
				\draw[brightmaroon, line width=2pt] (6pt, 95pt) -- (8pt, 95pt) -- (10pt, 98pt) -- (12pt, 103pt) -- (14pt, 106pt) -- (18pt, 106pt) -- (22pt, 119pt) -- (25pt, 135pt) -- (29pt, 145pt) -- (57pt, 148pt) -- (84pt, 158pt) -- (112pt, 180pt) -- (141pt, 201pt) -- (170pt, 206pt) -- (196pt, 203pt) -- (225pt, 209pt) -- (251pt, 230pt) -- (276pt, 264pt);
				
				\filldraw[color=black, fill=skyblue] (6pt, 34pt) circle (2pt);
				\filldraw[color=black, fill=skyblue] (8pt, 42pt) circle (2pt);
				\filldraw[color=black, fill=skyblue] (10pt, 53pt) circle (2pt);
				\filldraw[color=black, fill=skyblue] (12pt, 58pt) circle (2pt);
				\filldraw[color=black, fill=skyblue] (14pt, 58pt) circle (2pt);
				\filldraw[color=black, fill=skyblue] (18pt, 53pt) circle (2pt);
				\filldraw[color=black, fill=skyblue] (22pt, 48pt) circle (2pt);
				\filldraw[color=black, fill=skyblue] (25pt, 50pt) circle (2pt);
				\filldraw[color=black, fill=skyblue] (29pt, 55pt) circle (2pt);
				\filldraw[color=black, fill=skyblue] (57pt, 63pt) circle (2pt);
				\filldraw[color=black, fill=skyblue] (84pt, 77pt) circle (2pt);
				\filldraw[color=black, fill=skyblue] (112pt, 90pt) circle (2pt);
				\filldraw[color=black, fill=skyblue] (141pt, 95pt) circle (2pt);
				\filldraw[color=black, fill=skyblue] (170pt, 98pt) circle (2pt);
				\filldraw[color=black, fill=skyblue] (196pt, 100pt) circle (2pt);
				\filldraw[color=black, fill=skyblue] (225pt, 106pt) circle (2pt);
				\filldraw[color=black, fill=skyblue] (251pt, 122pt) circle (2pt);
				\filldraw[color=black, fill=skyblue] (276pt, 143pt) circle (2pt);
				
				\filldraw[color=black, fill=brightmaroon] (6pt, 95pt) circle (2pt);
				\filldraw[color=black, fill=brightmaroon] (8pt, 95pt) circle (2pt);
				\filldraw[color=black, fill=brightmaroon] (10pt, 98pt) circle (2pt);
				\filldraw[color=black, fill=brightmaroon] (12pt, 103pt) circle (2pt);
				\filldraw[color=black, fill=brightmaroon] (14pt, 106pt) circle (2pt);
				\filldraw[color=black, fill=brightmaroon] (18pt, 106pt) circle (2pt);
				\filldraw[color=black, fill=brightmaroon] (22pt, 119pt) circle (2pt);
				\filldraw[color=black, fill=brightmaroon] (25pt, 135pt) circle (2pt);
				\filldraw[color=black, fill=brightmaroon] (29pt, 145pt) circle (2pt);
				\filldraw[color=black, fill=brightmaroon] (57pt, 148pt) circle (2pt);
				\filldraw[color=black, fill=brightmaroon] (84pt, 158pt) circle (2pt);
				\filldraw[color=black, fill=brightmaroon] (112pt, 180pt) circle (2pt);
				\filldraw[color=black, fill=brightmaroon] (141pt, 201pt) circle (2pt);
				\filldraw[color=black, fill=brightmaroon] (170pt, 206pt) circle (2pt);
				\filldraw[color=black, fill=brightmaroon] (196pt, 203pt) circle (2pt);
				\filldraw[color=black, fill=brightmaroon] (225pt, 209pt) circle (2pt);
				\filldraw[color=black, fill=brightmaroon] (251pt, 230pt) circle (2pt);
				\filldraw[color=black, fill=brightmaroon] (276pt, 264pt) circle (2pt);
				
				\matrix [below left] at (current bounding box.north east) {
					\node [bigrednode, label=right:Without Skycube partial materialization] {}; \\
					\node [bigbluenode, label=right:Skyline concept lattice] {}; \\
				};
			\end{tikzpicture}
		}
		\caption{Queries' response time}\label{fig:temps_de_reponse_des_requetes}
	\end{figure}
	
	\begin{algorithm}[htbp]
		\caption{Sort-Filter-Skyline Algorithm(SFS)\label{algo:sfs}}
		\begin{algorithmic}[1]
			\Require~ \\
			The $r$ relation, sorted in ascending order by $\texttt{SUM}(t)$. \\
			The first tuple $t$ of $r$. \\
			The criterion set $\mathcal{C}$. \\
			\Ensure~ \\
			Skyline's tuple set $S_\mathcal{C}(r)$. \\
			\State $S_\mathcal{C}(r) := t$;
			\ForAll{$t' \in S_\mathcal{C}(r) \backslash  {t}$}
			\If{$\nexists t'' \in S_\mathcal{C}(r), t'' \prec t'$}
			\State $S_\mathcal{C}(r) := S_\mathcal{C}(r) \cup {t'}$;
			\EndIf
			\State \Return $S_\mathcal{C}(r)$
			\EndFor
		\end{algorithmic}
	\end{algorithm}
	
	In this experimentation, we want to compare a Skyline query's response time using the Sort-Filter-Skyline algorithm or SFS (algorithm~\ref{algo:sfs}), as presented by \cite{chomickiSkylinePresortingTheory2005}, first without Skycube partial materialization, and then with the Skyline concept lattice.
	
	SFS is the baseline algorithm for the experimentation. For each attribute combination from the subset lattice, a full, unoptimized run is performed.
	
	The SFS algorithm needs a monotonous or antimonotonous function (for example \texttt{SUM}) applied to the tuple set. Using this function, a topological ascending sorting is applied to $\texttt{SUM}(t)$. Let $t$ and $t'$ be two tuples of a $r$ relation with only positive values. Assuming all criteria shall be minimized. If $\texttt{SUM}(t') > \texttt{SUM}(t)$ then $t'$ does not dominate $t$.
	
	With the Skycube partial materialization, two cases are to be considered. In the simplest case, the query's result is included in the Skycube partial materialization and, consequently, the query's cost is also minimal. In the other case, the query's result must be computed from the associated closed Skycuboid using the theorem~\ref{theoreme_fondamental}. To improve the evaluation's quality and to provide meaningful results, each point of the figure~\ref{fig:temps_de_reponse_des_requetes} averages the response time of hundreds of queries with criteria cardinality increasing from 100 to 10000. Furthermore, to simulate user queries, we determine the criteria count for each query using a binomial distribution centred on the possible criteria count divided by two, then the criteria combination is randomly picked. This way, processed queries often factor several criteria, albeit not too much, so they have more meaning than if they were randomly picked. Results shows the efficiency of the Skyline concept lattice as it can quickly respond to any Skyline query and its gain factor is constant.
	
	\section{Conclusion and future work}
	
	In this paper, we presented the Skyline concept lattice, a constrained instance of a more general formal framework, the Agree concept lattice. This structure not only allows a Skycube lossless partial materialization but also an efficient computation of pruned Skycuboids. We can easily extend this materialization to $\epsilon$-Skylines, or approximate Skylines (\cite{papadiasProgressiveSkylineComputation2005}), by generalizing the definition of an Agree set by replacing strict equality by almost equality to the nearest $\epsilon$. Such an extension allows the decision maker to relax the dominance constraint when the result count is insufficient.
	
	Skycubes partially materialized that way, once merged, are used as input relation to compute the emerging Skycube or its reductions, the closed emerging Skycube, with measure data loss, and the closed emerging L-Skycube, without data loss. As long as the input format is respected, concepts are directly transposed from the emerging cube context to the emerging Skycube one. In the same vein, use of E-IDEA, F-IDEA and $\mathbb{C}$-IDEA algorithms from the IDEA platform is encouraged as a straightforward approach.
	
	Experimentations should be completed by a study on algorithm complexity for skycube computation and its variations. It will be object of future work.
	
	
	\bibliography{biblio}
	
	
	\begin{appendices}
	
	\section{Pokémon Showdown!}
	
	Pokémon Showdown! is an open source browser-based video game. It is a popular Pokémon\footnote{Pokémon, possibly because of its japanese origins, where the plural is not marked on substantives, is, usually, written identically regardless of the plural.} fighting simulator (with millions of monthly visitors and up to 20,000 simultaneous players) offering fully animated online Pokémon battles.
	
	It is used by Smogon University\footnote{Smogon University is a reference website, and a forum, about competitive Pokémon battles offering strategic guides.} since 2012. It was one of the first simulators to implement changes made to Pokémon versions Black 2 and White 2 and to include special formats proposing original approaches of Pokémon battles by tweaking the original video game's mechanics.
	
	\subsection{Pokémon Showdown! battle simulator}
	
	Pokémon Showdown! offers its users a system of random ranked battles in a desired format and generation of Pokémon. The website also features a chat room\footnote{A chat room is a virtual written conversation in real time, via an interposed screen on the Internet. By extension, a website offering this mode of communication is also known as a chatroom}, during battles.
	
	\begin{figure}[htbp]
		\begin{minipage}{\linewidth}
			\centering
			\includegraphics[width=\linewidth]{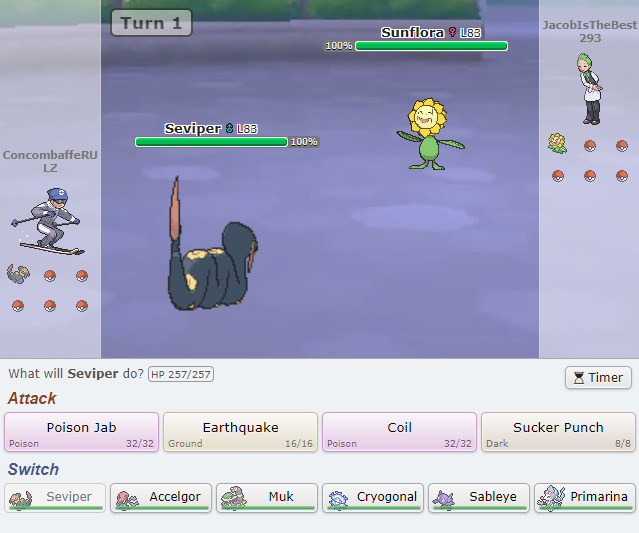}
			\caption[Pokémon Showdown! video game: battle example]{Pokémon Showdown! video game: battle example\footnote{Battle example in the video game Pokémon Showdown!}}\label{fig:jeu_videopokemon_showdown_combat}
		\end{minipage}
	\end{figure}
	
	It is possible to fight with randomly generated teams or by creating personal teams that are automatically validated according to the chosen format.
	
	\begin{figure}[htbp]
		\begin{minipage}{\linewidth}
			\centering
			\includegraphics[width=\linewidth]{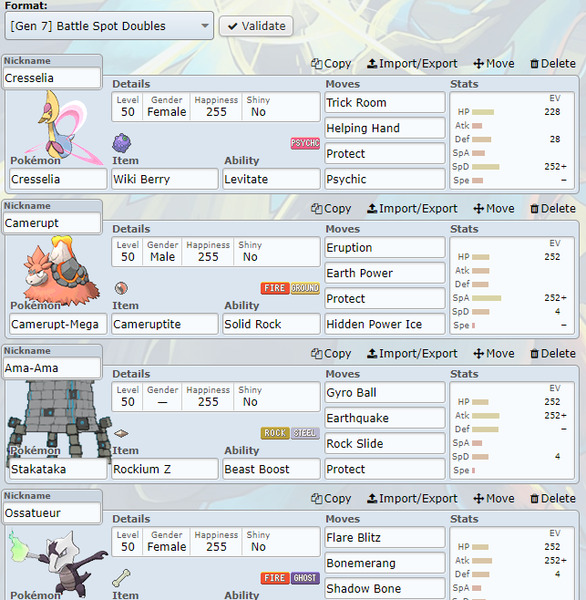}
			\caption[Pokémon Showdown! video game: team example]{Pokémon Showdown! video game: team example\footnote{Example of a team in the video game Pokémon Showdown!}}\label{fig:jeu_videopokemon_showdown_equipe}
		\end{minipage}
	\end{figure}
	
	It is possible to review our or other players' fights from an interface that allows you to consult recent fights or perform advanced searches by criteria.
	
	\subsection{Pokémon Showdown! ladder}\label{ssec:ladder_de_pokemon_showdown}
	
	Pokémon Showdown! features a point-based ranking system for each of the strategic tiers (cf. subsection~\ref{ssec:tier_strategique}), known as the ladder, or leaderboard. Players are ranked using the Elo, GXE (for glicko x-Act estimate) and Clicko-1 systems. The leaderboard for each ladder is available directly on Pokémon Showdown!
	
	\subsection{Pokédex}
	
	A kind of Pokémon encyclopedia, a Pokédex is a tool for researching Pokémon characteristics.
	
	Pokémon Showdown! offers its own Pokédex.
	
	
	\subsection{Individual value (IV)}\label{ssec:individual_value}\index{Individual value}\index{IV}
	
	An IV, or individual value, is a positive integer numerical value defined for each Pokémon characteristic. Most of the time, this data is not explicitly communicated in Pokémon games as such, since IV is a term in use by the community, "potential" is quoted instead. The higher the IV of a characteristic, the higher the corresponding characteristic.
	
	Let $C$ be a Pokémon characteristic, $N$ the value of the Pokémon nature ($1.1$, $0.9$ or $1$), $L$ the Pokémon level, $B$ the base value of the characteristic and $EV$\footnote{An EV, or effort value, or effort point, is a point a Pokémon can gain, which can be used to increase one of its characteristics.} the Pokémon effort points:
	\begin{itemize}
		\item here is the IV computing formula for the characteristic, other than HP\footnote{In Pokémon Showdown!, other Pokémon games and most of games, HP (for hit points) is used as a health-o-meter, or life-o-meter, of a character or for an item's toughness.}:
		\begin{equation*}
			IV = (\frac{C}{N} - 5) \times \frac{100}{L} - \frac{EV}{4} - 2 \times B
		\end{equation*}
		
		\item here is the IV computing formula for the HP characteristic:
		\begin{equation*}
			IV = \frac{100 \times (C - L - 10)}{L} - \frac{EV}{4} - 2 \times B
		\end{equation*}
	\end{itemize}
	
	\subsubsection{Strategic tier}\label{ssec:tier_strategique}\index{Strategic tier}
	
	\paragraph{Strategic dominance}\label{para:dominance_strategique}
	
	Strategic dominance, commonly called simply dominance, occurs when one strategy is better than another for one player, regardless of another player's strategy.
	
	\begin{definition}\index{Strategic dominance}
		Let $s_i$ and $s'_i$ be two strategies for a player $i$, and let $S_{-i}$ be the set of all possible strategies for other players:
		\begin{itemize}
			\item $s_i$ dominate strictly $s'_i$ if $\forall s_{-i} \in S_{-i}, u_i(s_i, s_{-i}) > u_i(s'_i, s_{-i})$;
			\item $s_i$ dominate weakly $s'_i$ if $\forall s_{-i} \in S_{-i}, u_i(s_i, s_{-i}) \geq u_i(s'_i, s_{-i})$.
		\end{itemize}
	\end{definition}
	
	\subsubsection{Game balancing}\label{sssec:equilibrage}
	
	Game balancing is a design phase about the tweaking of the gameplay to create an optimal player experience. This involves adjusting different elements of the game, such as game mechanics, rules and user interface elements, to ensure that the game is fun and engaging. This step is particularly important in video games, where maintaining the right balance is essential to ensure a satisfying gaming experience for players.
	
	To be optimal, a game must be fair to all players. This means that no one player should have an advantage over the others. To maintain balance, it is possible to create a symmetrical game by offering each player the same opportunities and resources. This ensures that no one has an unfair advantage over the others. However, many video games are inherently asymmetrical, which can lead to dominance. It comes down to balancing them, which is a recurring issue at the heart of video game designers' concerns. Balancing a game can be notably touchy, as each game is unique, with many different elements that need to be adjusted to create a balanced gaming experience.
	
	\subsubsection{Strategic tier in a video game}\label{sssec:tier_strategique_dans_un_jeu_video}
	
	A strategic tier is an unofficial categorization within a video game, based on strategic dominance, according to the performance or victory rates of characters, character teams or strategies employed.
	
	While unofficial and resulting from a video game's unbalanced nature, strategic tiers are fundamental to players and aim, paradoxically, to improve unbalances by grouping together characters, character teams or strategies considered "fair" by the player community or video game development teams.
	
	\subsubsection{Pokémon Showdown! strategic tiers}\label{sssec:tier_strategique_de_pokemon_showdown}
	
	In Pokémon Showdown!, a strategic tier is a hierarchical classification meant to sort Pokémon according to their efficiency, or more often, their IV (cf. subsection~\ref{ssec:individual_value}).
	
	In addition to the strategic tiers specific to Pokémon Showdown!, The Pokémon Company, managing the whole Pokémon license, has established its own strategic tiers, the Video Game Championship (VGC), which is used for qualifying tournaments for the annual Pokémon World Championships and other potential e-sport tournaments.
	
	In addition, the Limbo strategic tier is often used. This is a transitional strategic tier, introduced with the launch of a new generation of video games, adding new Pokémon requiring many battles in order to be empirically categorized. This strategic tier is not intended for use as a competitive format.
	
	Here are the main Pokémon Showdown! strategic tiers: 
	\begin{itemize}
		\item Anything goes: all Pokémon are allowed;
		\item Uber (U): Pokémon with a power or capacities considered as too unbalanced (banished from most competitions);
		\item Over Used (OU): the most used Pokémon throughout competitions;
		\item Under Used (UU): Pokémon juged too weak against OU Pokémons while being hardcore for most players;
		\item Rarely Used (RU): or Little Used (LU);
		\item Never Used (NU): someone played that, once;
		\item PU: Never Used (NU) of the Never Used (NU);
		\item Not Fully Evolved (NFE): Pokémon that have not reached their final evolution;
		\item Little Cup (LC): level 5 and first evolution stage (with evolution possibility) Pokémon only.
	\end{itemize}
	
	
	\section{Kaggle}
	
	Kaggle is a website, founded in 2010 and acquired by Google in 2017, offering machine learning\footnote{Machine learning is a field of study in artificial intelligence that uses mathematical and statistical approaches to "learning" computer systems (in other words, improving their algorithmic performance to meet needs without having been specifically developed to do so), usually from large datasets.} competitions in data science and graciously making datasets available to researchers to carry out data mining or machine learning projects.
	
	To complete the information required for our use case, and in particular to compute the \texttt{Rarity} (of the team's appearance sequence), which is one of the criteria of our use case, we relied on the Pokemon All Status Data (Gen1 to 9) resource provided on Kaggle by Takamasa Kato, a data science researcher.

	\end{appendices}
	
\end{document}